\documentclass[english,oneside]{article}

\input{./my_config.sty}
\usepackage{import}
\usepackage[a4paper, total={6.5in, 10in}]{geometry}
\usepackage{bbm}	
\usepackage{chngcntr} 
\usepackage{csquotes}
\usepackage{pdflscape}
\usepackage{orcidlink}

\input{./definitions.sty}

\bibliography{./zotero}

\usepackage[
	repeqn=independent,
	appendix=inline,     
	bibliography=common
]{apxproof}   

\graphicspath{ {figures/} }

\newtheoremrep{thm}[theorem]{Theorem}
\newtheoremrep{prop}[theorem]{Proposition}
\newtheoremrep{defi}[theorem]{Definition}
\newtheoremrep{rem}[theorem]{Remark}
\newtheoremrep{cor}[theorem]{Corollary}
\title{On the Estimation of Own Funds for Life Insurers:\\ A Study of Direct, Indirect, and Control Variate Methods\\ in a Risk-Neutral Pricing Framework}

\author{Mark-Oliver Wolf$^{1,2,}$\footnote{\href{mailto:mark-oliver.wolf@outlook.de}{mark-oliver.wolf@outlook.de}}\;\,\orcidlink{0000-0002-3698-9266}}

\date{%
	$^1$Fraunhofer Institute for Industrial Mathematics ITWM\\%
	$^2$University of Kaiserslautern-Landau RPTU\\[3ex]%
	\today
}

\begin{document}


\maketitle

\begin{abstract}
	The Solvency Capital Requirement ($\SCR$) calculation is computationally intensive, relying on the market-consistent estimation of own funds. While Solvency II prioritizes the direct valuation method, it theoretically yields the same value as the indirect method. This paper evaluates their practical performance within a risk-neutral pricing framework.

	First, we present a simplified proof that direct and indirect estimators converge to the same value. For $T$ being the number of time steps in the simulation, we then introduce a novel family of $2^T$ mixed estimators including both methods as edge cases, integrating them into a control variate framework for significant variance reduction. This framework is further extended to incorporate market frictions for real-world applicability.

	Evaluating these estimators on three life insurance asset-liability management models demonstrates that their performance is fundamentally driven by the degree of asset-liability coupling. While stronger coupling in realistic settings consistently favors the indirect method, neither baseline estimator is universally superior. Furthermore, this coupling directly impacts the success of our proposed control variates. They can dramatically reduce variance to one-tenth of the standard direct estimator, but their efficacy remains model-dependent.
	\newline\newline
	The source code is publicly available on the \href{https://gitlab.cc-asp.fraunhofer.de/itwm-fm-lv-public/wolf-estimation-of-own-funds}{Fraunhofer Gitlab}.

	\vspace{1em}
	\noindent\textbf{Keywords:} Solvency II, Own Funds, Available Capital, Asset-Liability Management, Variance Reduction, Risk-Neutral Pricing, Life Insurance.
\end{abstract}

\newpage
\newpage

\renewcommand{\arraystretch}{1.3}


\section{Introduction}\label{sec:introduction}

The Solvency~II directive (cf.\ \cite{eiopa_Directive2009138_2009}), officially implemented in January 2016, represents the prevailing prudential framework for insurance and reinsurance entities withing the European Union (cf.\ \cite{eiopa_SolvencyII_}).
Within this framework, the concept of \emph{own funds} defines the solvency capital requirement ($\SCR$) via its $99.5\%$ quantile, representing the regulatory capital that must be held as a buffer against potential losses. Per \cite[Article~88 and~77]{eiopa_Directive2009138_2009}, it is defined as `the excess of assets over liabilities, [based on a market-consistent valuation].'
Under the Swiss Solvency Test, this concept is referred to as \emph{risk-bearing capital} (cf.\ \cite{finma_SwissSolvencyTest_2024}).
Mathematically, it is sometimes called \emph{available capital} ($\AC$) to be independent of any changing regulatory frameworks. Hence, we are also adopting this term for the remainder of this article.

Naturally, the determination of the $\AC$ requires knowing the market-value of assets and liabilities. While in practice not all assets are marked-to-market, we assume for the sake of clarity that asset prices are observable and readily available without the need for simulation. In contrast, the exact market-value of liabilities is a priori unknown. Crucially, many insurance liabilities involve strongly path-dependent cash flows, driven by complex product features such as cliquet-style guarantees, dynamic profit-sharing rules, policyholder behavior (e.g., lapses), and various frictional costs. This inherent complexity renders closed-form valuation infeasible and forces practitioners to rely on Monte-Carlo schemes to project and discount the associated cash flows.

As mandated by regulatory frameworks, the $\AC$ has to be calculated at a risk horizon, typically one year into the future. Under a Monte-Carlo framework, this requires a primary simulation to project economic scenarios up to the risk horizon, and a secondary, nested simulation to estimate the $\AC$ via its conditional expectation. Evaluating complex, path-dependent cash flows within such a nested setup imposes an immense computational burden. Consequently, firms typically perform crude nested Monte Carlo estimation only for mandatory regulatory reporting. For more frequent, internal risk monitoring, identifying and applying effective variance reduction techniques is genuinely valuable, as it bridges the gap between theoretical solvency requirements and practical feasibility. To this end, the industry often employs approximation techniques like replicating portfolios (cf.\ \cite{beutner_FastConvergenceRegressLater_2013}, \cite{natolski_MathematicalFoundationReplicating_2018}, \cite{fernandez-arjona_MachineLearningApproach_2022}) or Least-Squares Monte Carlo (cf.\ \cite{bauer_SolvencyIINested_2009}, \cite{benedetti_CalculationRIskMeasures_2017}, \cite{krah_LeastSquaresMonteCarlo_2018}, \cite{ha_LeastsquaresMonteCarlo_2022}) to estimate the $\AC$ more efficiently (cf.\ \cite{pelsser_DifferenceLSMCReplicating_2016} for a comparison).

Under Solvency~II, the market-consistent value of liabilities is primarily given by the best-estimate liabilities (BEL), calculated as the probability‑weighted average of future cash flows discounted with the risk‑free term structure (\cite[Art.~75]{eiopa_Directive2009138_2009}). Following \cite{bauer_CalculationSolvencyCapital_2012a}, this corresponds to the \emph{direct method} (or option-pricing method), which determines value based on the expected discounted cash flows of the insurance obligations. In contrast, for internal capital management and investor communication, many insurers rely on embedded‑value frameworks (e.g., MCEV or EEV in the past), known as the \emph{indirect method} (or actuarial appraisal method), cf.\ \cite{americanacademyofactuaries_FairValuationInsurance_2002}. This approach assesses value through the lens of shareholder profitability, deriving the liability value by combining the net asset value with the present value of future profits (PVFP). While these approaches offer distinct perspectives, \cite{girard_MarketValueInsurance_2000} and \cite{girard_ApproachFairValuation_2002} have proven their theoretical equivalence, provided a consistent set of assumptions is applied.

\cite{bauer_CalculationSolvencyCapital_2012a} studied the general mathematical framework underlying the $\SCR$ computation in Solvency II. They also derived pricing formulas for the two methods using the no-arbitrage theory of pricing via equivalent martingale measures (also called risk-neutral pricing framework). Regarding these two methods, they argue:
\begin{displayquote}
	`While of course the quantity to be estimated is \textemdash{} or at least should be \textemdash{} the same for both procedures, the two methods may well yield different estimators for the $\AC$ and, hence, for the SCR. [\dots] In particular, our numerical experiments illustrate that the quality of the resulting estimates can differ significantly.'
\end{displayquote}

Their research motivates several key questions:
\begin{itemize}
	\item What is the key condition such that direct and indirect estimator converge to the same quantity?
	\item How can we validate a correct implementation of either method?
	\item Are these two methods the only possible estimators?
	\item Can the direct and indirect estimators be combined within a control variate framework to produce a more robust and efficient estimate?
	\item Which estimator is more computationally efficient for a given model, and what determines its superiority?
\end{itemize}

In this paper, we answer these questions with a novel view and improved estimators for calculating the $\AC$.
With minimal assumptions, we present a straightforward proof that the direct and indirect estimator converge to the same capital in \Cref{subsec:fair_value_of_assets}.
We present a simplified example in \Cref{rem:edge_cases_of_estimators_variance} highlighting that the superiority of either estimator (i.e., having less variance) is a priori unclear.
Surprisingly, the mathematical proof of their equality relies only on the no-arbitrage valuation of the asset process itself, regardless of the complexity of the rules governing the liability cash flows.
Because of this equality, the indirect method can be used in practice as a check for the direct method's correct implementation and vice versa.
The proof highlights a key equality based on the tower property of conditional expectations, allowing us to construct a new family of \emph{mixed estimators} consisting of $2^{\text{number of time steps}}$ distinct estimators for the $\AC$ in \Cref{subsec:mixed_estimator}. We show that direct and indirect estimator are special cases of mixed estimators.
The equal expectation of all estimators immediately implies that their difference has a known expectation of zero.
We leverage this insight to construct single- and multi-control variate estimators for the $\AC$ in \Cref{sec:control_variates}, where we plug our novel estimators into the well known existing control variate framework to enable an effective variance reduction. In \Cref{sec:leakage}, we extend the model to include frictional costs.

The performance of the estimators is evaluated in \Cref{sec:numerical_study} using three benchmark models representing life insurance companies.
We first consider two variants of the model developed by \cite{bauer_RiskneutralValuationParticipating_2006}: the MUST case, representing a mandatory policyholder participation scheme, and the IS case, a more realistic target-based approach.
For a more comprehensive setting, we then employ the openIRM model by \cite{wolf_OpenIRMPubliclyAccessible_2025}, an artificial internal risk model designed to be more representative of those used in practice.
The results indicate that estimator efficiency is highly sensitive to the model's parameterization. These findings reinforce the intuition from \Cref{rem:edge_cases_of_estimators_variance}: the relative variance of either method depends on the degree of coupling between assets and liabilities. Specifically, the indirect estimator becomes more efficient as policyholder participation in market dynamics increases.

In more realistic setups, the indirect estimator is often the preferred choice. Furthermore, the control variate approach combining the direct and indirect methods dramatically outperforms individual estimators in complex configurations, though its gains are less pronounced when policyholders have a smaller share of market returns.

The remainder of this article is organized as follows. \Cref{sec:underlying_model} establishes the mathematical framework and the risk-neutrality assumption. \Cref{sec:direct_and_indirect_method} presents the economic derivation of both valuation methods to clarify their conceptual origins. Finally, \Cref{sec:conclusion} summarizes our findings and outlines directions for future research.

\section{Preliminaries}\label{sec:preliminaries}

\subsection{Mathematical Framework}\label{sec:underlying_model}

The presented model is an abstraction of the typical asset-liability management models used, shown in for example \cite{bauer_RiskneutralValuationParticipating_2006}, \cite{bauer_CalculationSolvencyCapital_2012a}, and \cite{wolf_OpenIRMPubliclyAccessible_2025}.
We model the asset-liability management of a life insurance company for $N$ equidistant time steps $t = 0, \dt, 2\dt, \dots, \timehorizon \dt$, where $\dt$ is the length of each step in years and $\timehorizon \dt$ is the final simulation time. We assume that $1 / \dt \in \mathbb{N}$, i.e., full years fall on the time grid. For simplicity, we denote the time directly via its index $t = 0, 1, \dots, \timehorizon$.
Let $(\Omega, \cF, \Prob, (\cF_t)_{t = 0, 1, \dots, \timehorizon})$ be a complete filtered probability space on which all relevant quantities exist, where $\Omega$ denotes the space of all possible states of the financial market and $\Prob$ is the physical (real-world) probability measure. The $\sigma$-algebra $\cF_t$ represents all information about the market up to time $t$, and the filtration $(\cF_t)_{t = 0, 1, \dots, \timehorizon}$ is assumed to satisfy the usual conditions.

\begin{defi}[Abstract ALM model]
	At the beginning of the simulation, our company starts with initial assets $\assetsZero$ and initial liabilities $\liabsZero$.
	All of our assets are invested in a financial market, so for each time step we get a market return $\assetsReturn{t}$. For $t > 0$, we get
	\begin{align}
		\assetsBefore{t} \coloneqq \assetsAfter{t-1} \assetsReturn{t}.
	\end{align}
	For each time $t$, the assets generate a shareholder cash flow $\shcf{t}$ and a policyholder cash flow $\phcf{t}$.
	For $t>0$, we get
	\begin{align}\label{eq:assets_cash_flows}
		\assetsAfter{t} \coloneqq \assetsBefore{t} - \shcf{t} - \phcf{t}.
	\end{align}
	The company's assets are always valued mark-to-market in this model.

	Similarly to the assets, the company's liabilities are described by the processes $\liabsBefore{t}$ and $\liabsAfter{t}$. A more detailed description of their dynamics is not needed here.
	We define the free fund before and after the cash flows as $\freefundBefore{t} \coloneqq \assetsBefore{t} - \liabsBefore{t}$ and $\freefundAfter{t} \coloneqq \assetsAfter{t} - \liabsAfter{t}$, respectively.
\end{defi}

\begin{defi}[Risk-neutral Abstract ALM model]
	We call an abstract ALM model \emph{risk-neutral}, if there exists an equivalent martingale measure $\mQ$ with numeraire $\numeraire{t}$ such that
	\begin{align}\label{eq:martingale_property}
		\E^{\mQ}*{\frac{\numeraire{t}}{\numeraire{t-1}} \assetsBefore{t} | \cF_{t-1}} = \assetsAfter{t-1},
		\hspace{1cm}
		\text{or equivalently}
		\hspace{1cm}
		\E^{\mQ}*{\frac{\numeraire{t}}{\numeraire{t-1}} \assetsReturn{t} | \cF_{t-1}} = 1.
	\end{align}
	We define the abbreviated notation $\disc{t_1}{t_2} \coloneqq \numeraire{t_2} / \numeraire{t_1}$.
\end{defi}

\subsection{Definition of direct and indirect method}\label{sec:direct_and_indirect_method}

By \cite{bauer_CalculationSolvencyCapital_2012a}, the available capital corresponds to the amount of available financial resources that can serve as a buffer against risks and absorb financial losses.
Under the regulatory reporting of Solvency II, it has to be computed by studying the policyholders' future cash flows, which corresponds to the \emph{direct method}.
An alternative method for internal reporting is the \emph{indirect method}, where the shareholders' cash flows are observed.
It is similar to the market-consistent embedded value ($\MCEV$), which is
`[...] a measure of the consolidated value of shareholders' interest in the covered business' (cf.\ \cite{cfoforum_EuropeanEmbeddedValue_2016}).
By \cite{deloitteemeaa&is_FinancialKPIsIFRS_2023}, the MCEV is less relevant for regulatory purposes in countries that are subject to Solvency II, but still commonly used in other countries like the Swiss market.
We will denote the available capital computed by the direct and indirect method as $\ACdir$ and $\ACind$, respectively.

\cite{bauer_CalculationSolvencyCapital_2012a} derived that in Solvency II, the methods can be written for a risk horizon of~$\tau$ representing $1$ year as
\begin{align}
	\begin{split}
		\ACdir_1 &= \MVA_1 - \MVL_1 \\
		&= \MVA_1 - \E^\mQ*{\substack{\text{cash flow to}\\ \text{policyholders}} | \cF_1},
	\end{split} \\
	\begin{split}
		\ACind_1 &= \MCEV_1 \\
		&= \ANAV_1 + \PVFP_1 - \CoC_1 \\
		&\;\hat{=} \E^\mQ*{\substack{\text{cash flow to}\\ \text{shareholders}} | \cF_1},
	\end{split}
\end{align}
where $\MVA$ and $\MVL$ are the market-value of assets and liabilities, respectively. The present value of future profits ($\PVFP$) is defined as the expected (time value-adjusted) return on investment for shareholders, the cost-of-capital ($\CoC$) is assumed to be 0 and the adjusted net asset value ($\ANAV$) is assumed to be the net asset value ($\NAV$).

We can translate this definition to our abstract ALM model.

\begin{defi}
	In an abstract asset-liability management (ALM) model, the direct and indirect available capital at time $\riskhorizon$ and with time horizon $\timehorizon>\riskhorizon$ are defined as:
	\begin{align}
		\ACdir_\riskhorizon(\timehorizon) &\coloneqq \assetsAfter{\riskhorizon} - \E^\mQ[\bigg]{
			\sum_{t = \riskhorizon + 1}^{\timehorizon} \disc{\riskhorizon}{t} \phcf{t}	+ \disc{\riskhorizon}{\timehorizon} \liabsAfter{\timehorizon}
			| \cF_\riskhorizon
		}, \\
		\ACind_\riskhorizon(\timehorizon) &\coloneqq \E^\mQ[\bigg]{
			\sum_{t = \riskhorizon + 1}^{\timehorizon} \disc{\riskhorizon}{t} \shcf{t}
			+ \disc{\riskhorizon}{\timehorizon} \freefundAfter{\timehorizon}
			| \cF_\riskhorizon
		}.
	\end{align}

	In practice, the available capital has to be estimated using observations produced by a simulation model.
	For this purpose, we define the direct and indirect present value $\PV$ at time $\riskhorizon$ as
	\begin{align}
		\PVdir_\riskhorizon(\timehorizon)
		&\coloneqq
		\sum_{t = \riskhorizon + 1}^{\timehorizon} \disc{\riskhorizon}{t} \phcf{t}
		+ \disc{\riskhorizon}{\timehorizon} \liabsAfter{\timehorizon}, \\
		\PVind_\riskhorizon(\timehorizon)
		&\coloneqq
		\sum_{t = \riskhorizon + 1}^{\timehorizon} \disc{\riskhorizon}{t} \shcf{t}
		+ \disc{\riskhorizon}{\timehorizon} \freefundAfter{\timehorizon}.
	\end{align}
\end{defi}

By sampling from their respective distribution conditioned on $\cF_\tau$, the risk model produces samples $\PVdir_{\riskhorizon, j}$ and $\PVind_{\riskhorizon, j}$.
Then, for $\methodchoice$ being the chosen method, a naive Monte Carlo estimation with $\nInner$ realizations results in the following estimator:
\begin{align}
	\ACest_\riskhorizon^{\methodchoice}
	\coloneqq
	\frac{1}{\nInner} \sum_{j=1}^{\nInner} \PV_{\riskhorizon, j}^{\methodchoice},
\end{align}
where of course $\E^{}*{\ACest_\riskhorizon^{\methodchoice} | \cF_\tau} = \AC_\riskhorizon^{\methodchoice}(T)$.

Ultimately, the $\SCR$ is then the $99.5\%$ percentile over a sample of observations $\ACest_{\riskhorizon, i}^{\mathrm{M}}, i=1,\dots,n$ for varying risk factors expressed in $\cF_{\riskhorizon, i}$.

\section{Uniqueness of available capital and alternative estimators}\label{sec:uniqueness_of_ac_and_mixed_estimators}

\subsection{Fair value of assets and equal expectation of direct and indirect estimator}\label{subsec:fair_value_of_assets}
For the derivation of available capital estimators, the fair value of assets plays a major role.
Before presenting a generalized proof in \Cref{cor:equality_of_mixed_ac}, this section establishes the equality of the methods in \Cref{thm:equality_dir_ind}. This initial result serves a didactic purpose, providing a clear insight into the fundamental reason for the equality without the complexities of the general case for mixed estimators.
We start by proving a fundamental equality for an ALM model with a single cash flow that simplifies the general idea. Afterwards, we can quickly derive the equalities for the available capital we are interested in.
This first corollary states the fundamental principle of asset valuation in a risk-neutral framework. It confirms that the value of the asset at any time is precisely the risk-neutral expected present value of all its subsequent cash flows and its terminal value.

\begin{cor}\label[cor]{cor:fair_price_of_asset}
	Let $\assetsBefore{t} = \assetsAfter{t-1} \assetsReturn{t}$ and $\assetsReturn{t}$ have the risk-neutral market property specified in \eqref{eq:martingale_property}. Furthermore, let $\assetsAfter{t} = \assetsBefore{t} - \cf{t}$ for some arbitrary cash flow $\cf{t}$ and for all $t$.
	Then, for all $\riskhorizon \in \{0, 1, 2, \dots, \timehorizon\}$, we have
	\begin{equation}
		\assetsAfter{\riskhorizon}
		=
		\E^{\mQ}*{
			\disc{\riskhorizon}{\timehorizon} \assetsAfter{\timehorizon}
			+ \sum_{t=\riskhorizon+1}^{\timehorizon} \disc{\riskhorizon}{t} \cf{t}
			| \cF_\riskhorizon
		}.
	\end{equation}
\end{cor}
\begin{proof}
	We prove the statement with an induction over the time horizon $\timehorizon$. We begin with $\timehorizon = \riskhorizon + 1$. Then
	\begin{equation}
		\assetsAfter{\riskhorizon+1} = \assetsAfter{\riskhorizon} \assetsReturn{\riskhorizon+1} - \cf{\riskhorizon+1},
	\end{equation}
	hence we can apply the risk-neutral property \eqref{eq:martingale_property} and get
	\begin{align}
		\E^{\mQ}*{
			\disc{\riskhorizon}{\riskhorizon+1} \left( \assetsAfter{\riskhorizon+1} + \cf{\riskhorizon+1} \right)
			| \cF_\riskhorizon
		}
		= \assetsAfter{\riskhorizon}.
	\end{align}

	We proceed with the induction step, $\timehorizon-1 \to \timehorizon$. We apply the tower property for conditional expectations to get
	\begin{align}\label{eq:proof_eq_1}
	\begin{split}
		\E^{\mQ}*{
			\disc{\riskhorizon}{\timehorizon} \assetsAfter{\timehorizon}
			| \cF_\riskhorizon
		}
		&=
		\E^{\mQ}*{
			\disc{\riskhorizon}{\timehorizon-1} \disc{\timehorizon-1}{\timehorizon} \left( \assetsAfter{\timehorizon-1} \assetsReturn{\timehorizon} - \cf{\timehorizon} \right)
			| \cF_\riskhorizon
		} \\
		&=
		\E^{\mQ}*{
			\disc{t}{\timehorizon-1}
			\E^{\mQ}*{
				\disc{\timehorizon-1}{\timehorizon} \left( \assetsAfter{\timehorizon-1} \assetsReturn{\timehorizon} - \cf{\timehorizon} \right)
				| \cF_{\timehorizon-1}
			}
			| \cF_\riskhorizon
		} \\
		&=
		\E^{\mQ}*{
			\disc{\riskhorizon}{\timehorizon-1} \assetsAfter{\timehorizon-1}
			- \disc{\riskhorizon}{\timehorizon} \cf{\timehorizon}
			| \cF_\riskhorizon
		}.
	\end{split}
	\end{align}

	Ultimately, we use \eqref{eq:proof_eq_1} and the induction hypothesis for $\timehorizon-1$ to finish the proof via
	\begin{align}
	\begin{split}
		&\E^{\mQ}*{
			\disc{\riskhorizon}{\timehorizon} \assetsAfter{\timehorizon}
			+ \sum_{t=\riskhorizon+1}^{\timehorizon} \disc{\riskhorizon}{t} \cf{t}
			| \cF_\riskhorizon
		} \\
		&=
		\E^{\mQ}*{
			\disc{\riskhorizon}{\timehorizon-1} \assetsAfter{\timehorizon-1}
			+ \sum_{t=\riskhorizon+1}^{\timehorizon-1} \disc{\riskhorizon}{t} \cf{t}
			| \cF_\riskhorizon
		} \\
		&\stackrel{\text{(IH)}}{=}
		\assetsAfter{\riskhorizon}.
	\end{split}
	\end{align}
\end{proof}

The next remark follows directly from \Cref{cor:fair_price_of_asset}.

\begin{rem}\label[rem]{rem:equality_of_assets_equation}
	Setting $\cf{t} = \cf{t}^A + \cf{t}^B$, we can rewrite the identity of \Cref{cor:fair_price_of_asset} to get
	\begin{align}
	\begin{split}
		\assetsAfter{\riskhorizon}
		- \E^{\mQ}*{
			\sum_{t=\riskhorizon+1}^{\timehorizon} \disc{\riskhorizon}{t}
			\cf{t}^A
			| \cF_\riskhorizon
		}
		=
		\E^{\mQ}*{
			\sum_{t=\riskhorizon+1}^{\timehorizon} \disc{\riskhorizon}{t}
			\cf{t}^B
			+ \disc{\riskhorizon}{\timehorizon} \assetsAfter{\timehorizon}
			| \cF_\riskhorizon
		}
	\end{split}
	\end{align}
\end{rem}

Thus, the equality of the direct and indirect available capital follows directly from the preceding remark.
\begin{thm}\label[thm]{thm:equality_dir_ind}
	In a risk-neutral abstract ALM model, the direct available capital equals the indirect available capital for all $0 \leq \riskhorizon < \timehorizon$, i.e.,
	\begin{align}
		\ACdir_\riskhorizon(\timehorizon)
		=
		\ACind_\riskhorizon(\timehorizon),
	\end{align}
	where
	\begin{align}
		\ACdir_\riskhorizon(\timehorizon) &= \assetsAfter{\riskhorizon} - \E^\mQ[\bigg]{
			\sum_{t = \riskhorizon + 1}^{\timehorizon} \disc{\riskhorizon}{t} \phcf{t}	+ \disc{\riskhorizon}{\timehorizon} \liabsAfter{\timehorizon}
			| \cF_\riskhorizon
		}, \label{eq:direct_method} \\
		\ACind_\riskhorizon(\timehorizon) &= \E^\mQ[\bigg]{
			\sum_{t = \riskhorizon + 1}^{\timehorizon} \disc{\riskhorizon}{t} \shcf{t}
			+ \disc{\riskhorizon}{\timehorizon} \freefundAfter{\timehorizon}
			| \cF_\riskhorizon
		}.\label{eq:indirect_method}
	\end{align}
\end{thm}
\begin{proof}
	Set $\cf{t}^A = \phcf{t}$ and $\cf{t}^B = \shcf{t}$, recall that $\freefundAfter{t} = \assetsAfter{t} - \liabsAfter{t}$, and apply \Cref{rem:equality_of_assets_equation}.
\end{proof}

In summary, we have proved that the direct and indirect estimators share the same expectation: the available capital. This result makes the term `available capital' unambiguous, regardless of the method used. However, the estimators themselves are not equivalent. To illustrate the ambiguity in choosing between them, we begin with a stylized construction that illustrates why neither estimator is a priori superior. We then proceed in \Cref{sec:numerical_study} to evaluate their performance in more empirically relevant life insurance models.

\begin{remark}\label[remark]{rem:edge_cases_of_estimators_variance}
	Assume a risk-neutral abstract ALM model. Assume that $\phcf{t} = \shcf{t} = 0$ for all $t$, i.e., the model does not pay out any cash flows during the simulation. Then $\assets{t}{} \coloneqq \assetsAfter{t} = \assetsBefore{t}$ and $\liabs{t}{} \coloneqq \liabsAfter{t} = \liabsBefore{t}$. Let the discount factor $\disc{\tau}{t}$ be non-stochastic for all $t$. We can observe:
	\begin{itemize}
		\item If $\liabs{t}{}$ is non-stochastic, e.g., if the guaranteed interest rate determining the policyholder's returns is higher than any observed market returns, then $\Var( \ACdirest ) = 0$ and $\Var( \ACindest ) > 0$.
		\item If $\liabs{t}{} = \assets{t}{}$ for all $t$, e.g., if the full market return is passed on to the policyholders, then $\Var( \ACdirest ) > 0$ and $\Var( \ACindest ) = 0$.
	\end{itemize}

	While these idealized cases are not practically viable, they represent two extremes where one method is clearly superior. In practice, firms operate between these poles, making the optimal choice of estimator a priori unclear.
\end{remark}


\subsection{The mixed estimator}\label{subsec:mixed_estimator}

Using the same trick of risk-neutral equality from \Cref{cor:fair_price_of_asset}, we can write down the family of mixed estimators that also have the same expectation but slightly differ in distribution compared to the direct and indirect estimator.

\begin{defi}
	In an abstract ALM model, we define the mixed estimator for the available capital and correspondingly the mixed present value for a subset of time steps $\timeInMixed \subseteq \{1, 2, \dots, \timehorizon \}$ as
	\begin{align}
		\begin{split}
		\ACmix_{\riskhorizon}(\timehorizon; \timeInMixed)
		\coloneqq&
		\assetsAfter{\riskhorizon}
		- \E^{\mQ}*{
			\sum_{t \not\in \timeInMixed} \left[ \disc{\riskhorizon}{t} \phcf{t} \right]
			+ \disc{\riskhorizon}{\timehorizon} \liabsAfter{\timehorizon}
			| \cF_\riskhorizon
		} \\
		&- \E^\mQ[\bigg]{
			\sum_{t \in \timeInMixed} \left[ \disc{\riskhorizon}{t} \left( - \assetsAfter{t} - \shcf{t} \right)
			+ \disc{\riskhorizon}{t-1} \assetsAfter{t-1} \right]
			| \cF_\riskhorizon
		},
		\end{split} \\
		\PVmix_{\riskhorizon}(T; \timeInMixed)
		\coloneqq&
		\assetsAfter{\riskhorizon}
		- \sum_{t \not\in \timeInMixed} \left[ \disc{\riskhorizon}{t} \phcf{t} \right]
		+ \disc{\riskhorizon}{\timehorizon} \liabsAfter{\timehorizon}
		- \sum_{t \in \timeInMixed} \left[ \disc{\riskhorizon}{t} \left( - \assetsAfter{t} - \shcf{t} \right)
		+ \disc{\riskhorizon}{t-1} \assetsAfter{t-1} \right]. \label{eq:pvmix_def}
	\end{align}
\end{defi}

\begin{cor}
	Let $\timeInMixed_{\mathrm{all}} = \{1,2,\dots,\timehorizon\}$ and $\varnothing$ be the empty set.
	In a risk-neutral abstract ALM model, the following equalities hold:
	\begin{align}
		\PVmix_{\riskhorizon}(\timehorizon; \varnothing)
		&=
		\PVdir_{\riskhorizon}(\timehorizon), \\
		\PVmix_{\riskhorizon}(\timehorizon; \timeInMixed_{\mathrm{all}})
		&=
		\PVind_{\riskhorizon}(\timehorizon).
	\end{align}
\end{cor}
\begin{proof}
	For $\timeInMixed = \varnothing$, the second sum in \eqref{eq:pvmix_def} vanishes. The remainder is then identical to $\PVdir_\riskhorizon(\timehorizon)$. For $\timeInMixed = \timeInMixed_{\mathrm{all}}$, the first sum vanishes. We are left with
	\begin{align}
		\PVmix_{\riskhorizon}(\timehorizon; \timeInMixed_{\mathrm{all}})
		&=
		\assetsAfter{\riskhorizon} + \disc{\riskhorizon}{\timehorizon} \liabsAfter{\timehorizon}
		- \left[ \sum_{t=\riskhorizon+1}^{\timehorizon} \disc{\riskhorizon}{t} \left( - \assetsAfter{t} - \shcf{t} \right)
		+ \disc{\riskhorizon}{t-1} \assetsAfter{t-1} \right] \\
		&=
		\assetsAfter{\riskhorizon} + \disc{\riskhorizon}{\timehorizon} \liabsAfter{\timehorizon}
		- \left[
			\assetsAfter{\riskhorizon}
			- \sum_{t=\riskhorizon+1}^{\timehorizon} \disc{\riskhorizon}{\timehorizon} \shcf{t}
			- \disc{\riskhorizon}{\timehorizon} \assetsAfter{\timehorizon}
		\right] \\
		&= \PVind_{\riskhorizon}(\timehorizon)
	\end{align}
\end{proof}

\begin{cor}\label[cor]{cor:equality_of_mixed_ac}
	In a risk-neutral abstract ALM model, it holds for all $\timeInMixed \subseteq \{1, 2, \dots, \timehorizon \}$ that
	\begin{align}
		\ACmix_{\riskhorizon}(T; \timeInMixed)
		= \ACdir_{\riskhorizon}(T)
		= \ACind_{\riskhorizon}(T).
	\end{align}
\end{cor}
\begin{proof}
	We are done if we show that the direct and mixed available capital are equal. The direct available capital is defined as
	\begin{align}
	\begin{split}
		\ACdir_\riskhorizon(\timehorizon) = \assetsAfter{\riskhorizon} - \E^\mQ[\bigg]{
			\sum_{t = \riskhorizon + 1}^{\timehorizon} \disc{\riskhorizon}{t} \phcf{t}	+ \disc{\riskhorizon}{\timehorizon} \liabsAfter{\timehorizon}
			| \cF_\riskhorizon
		}.
	\end{split}
	\end{align}
	We use $\assetsAfter{t} = \assetsBefore{t} - \shcf{t} - \phcf{t}$ \eqref{eq:assets_cash_flows} to rewrite the expected discounted policyholder cash flow for all $t$ via
	\begin{align}
	\begin{split}
		\E^{\mQ}*{\disc{\riskhorizon}{t} \phcf{t}}
		&=
		\E^{\mQ}*{
			\disc{\riskhorizon}{t}
			(-\assetsAfter{t} - \shcf{t} + \assetsAfter{t-1} \assetsReturn{t})
			| \cF_\riskhorizon
		} \\
		&=
		\E^{\mQ}*{
			\disc{\riskhorizon}{t}
			(-\assetsAfter{t} - \shcf{t})
			+ \disc{\riskhorizon}{t-1} \disc{t-1}{t} \assetsAfter{t-1} \assetsReturn{t}
			| \cF_\riskhorizon
		}.
	\end{split}
	\end{align}
	Again, we use the tower property and the risk-neutral property of the assets to get
	\begin{align}
	\begin{split}
		\E^{\mQ}*{
			\disc{\riskhorizon}{t-1} \disc{t-1}{t} \assetsAfter{t-1} \assetsReturn{t}
			| \cF_\riskhorizon
		}
		&=
		\E^{\mQ}*{
			\disc{\riskhorizon}{t-1}
			\E^{\mQ}*{
				\disc{t-1}{t} \assetsAfter{t-1} \assetsReturn{t}
				| \cF_{t-1}
			}
			| \cF_\riskhorizon
		} \\
		&=
		\E^{\mQ}*{
			\disc{\riskhorizon}{t-1} \assetsAfter{t-1}
			| \cF_\riskhorizon
		}.
	\end{split}
	\end{align}
	Starting from the direct available capital, we can now just apply the above identity to $\phcf{t}$ for all $t \in \timeInMixed$ and get the mixed representation.
\end{proof}

\section{Control variates for estimating the available capital}\label{sec:control_variates}

In this section, we briefly introduce the concept of control variates for a single and multiple controls for our use case. We follow the presentation by \cite{glasserman_MonteCarloMethods_2010} closely, using random variables and unknown statistics for the single control to outline the idea, and observations and estimators for the multiple controls case as the version that can be applied in practice. Assume for this presentation that the expectations are appropriately conditioned if necessary.

\paragraph{Single control.}
For a given target variable $\PV$, our objective is the estimation of its expectation, the available capital.
If we have a second variable $\PV^c$ with the same expectation, $\E^{}*{\PV - \PV^c} = 0$ holds.
We define this difference as
\begin{equation}
	C \coloneqq \PV - \PV^c.
\end{equation}

Because the expectation of $C$ is known, we may now use it to construct a control variate estimator. Let $b \in \mathbb{R}$ be fixed, then the control variate present value is defined as
\begin{align}
	\PVcv(b)
	&\coloneqq \PV - b (C - \E^{}*{C}) \\
	&= \PV - b (\PV - \PV^c).
\end{align}
We know that for any $b$ we have $\E^{}*{\PVcv(b)} = \AC$. The coefficient $b$ is estimated by minimizing the variance of $\PVcv(b)$ and inserting the respective estimators for the variance. Note that, $b(C - \E^{}*{C})$ can be thought of as the projection of $\PV - \AC$ onto $C - \E^{}*{C}$, i.e., the remaining randomness of the control variate estimator is exactly the part that is orthogonal (in this case uncorrelated) to the chosen control.

\paragraph{Multiple controls.}
For a vector of controls $C_i = (C_i^{(1)}, \dots, C_i^{(d)})$ with known vector of expectations and $i=1,\dots,n$ realizations, we can write the estimated covariance matrix of target variable and controls as
\begin{align}
	\begin{pmatrix}
		S_C & S_{C, \PV} \\
		S_{C, \PV}^T & \hat{\sigma}_{\PV}
	\end{pmatrix},
\end{align}
where $S_C$ denotes the sample covariance matrix of the controls, $\hat{\sigma}_{\PV}$ denotes the sample covariance of the target variable and $S_{C, \PV}$ is the remaining cross-covariance vector.
For fixed $b \in \mathbb{R}^d$, the control variate estimator is
\begin{align}
	\overline{\PV}(b)
	=
	\overline{\PV} - b^T (\overline{\PV} - \overline{\PV^c}),
\end{align}
where the bar indicates the sample mean. Note that later, these $\overline{\PV}$ and $\overline{\PV^c}$ will be using the same underlying realizations of random variables, i.e., $\PV_i$ and $\PV^c_i$ will be computed from the same underlying simulation.
The estimator for the optimal $b^*$ is then given by
\begin{align}
	\hat{b} = S_C^{-1} S_{C, \PV}.
\end{align}
Two options are popular to efficiently estimate $b$. Either we use all of the data and then use the same data to estimate the final present value, thereby introducing a small bias in the estimation, or we split our data into a part for estimating $b^*$ once, and then use all other data to estimate the present value.
Note that, the above equations highlight that the control variate estimator is essentially a regression over the controls.
The resulting control variate estimator we employ is given by $\PVcv(\hat{b})$.
For $R^2$ being the coefficient of determination known from regression,
\begin{align}
	R^2 = \frac{ S_{C, \PV}^T S_{C}^{-1} S_{C, \PV} } { \hat{\sigma}_{\PV}^2 },
\end{align}
the $\PVcv(\hat{b})$ has a variance of $(1 - R^2) \sigma_{\PV}^2$. Hence, we can assess the decrease in variance by evaluating $R^2$. We will need this in \Cref{sec:numerical_study}.

%
%
%

\section{Extension to models with frictional costs}
\label{sec:leakage}

The framework developed thus far assumes that all cash flows originating from the asset portfolio are fully allocated to either policyholders by $\phcf{t}$ or shareholders by $\shcf{t}$. However, as heavily emphasized in the literature, \emph{frictional costs} play a crucial role in insurance markets, fundamentally driving capital management and pricing decisions (see, e.g., \cite{bauer_FinancialPricingInsurance_2013}, \cite{albrecher_AssetliabilityManagementLongterm_2018}). The economic impact of these market imperfections is particularly pronounced when covering long-term insurance liabilities (cf.\ \cite{albrecher_AssetliabilityManagementLongterm_2018}), as the necessity to hold capital over extended horizons amplifies the loss of frictions. These frictions can arise from various sources, such as corporate taxes, agency costs, investment management fees, or administrative burdens (cf.\ \cite[Sec.~2]{americanacademyofactuaries_FairValuationInsurance_2002}), and describe all costs that are directly deducted from the asset base and not paid out to either policyholders or shareholders. In this section, we demonstrate how our framework can be seamlessly extended to incorporate this economically important feature, thereby preserving the equality of the estimators and the validity of the proposed variance reduction techniques.

Let us introduce $\cfleak{t}$ as the cash flow representing frictions at time $t$. The asset dynamics after cash flows, as described in \eqref{eq:assets_cash_flows}, must be updated to:
\begin{equation}
	\label{eq:asset_leakage}
	\assetsAfter{t} = \assetsBefore{t} - \phcf{t} - \shcf{t} - \cfleak{t}.
\end{equation}
This impacts the relationship between the direct and indirect estimators as originally defined. If we retrace the proof of \Cref{thm:equality_dir_ind}, the fundamental asset valuation principle from \Cref{cor:fair_price_of_asset} now accounts for the total cash outflow $\cf{t} = \phcf{t} + \shcf{t} + \cfleak{t}$. This leads to the following relationship between the original estimators:
\begin{equation}
	\ACdir_\riskhorizon(\timehorizon) = \ACind_\riskhorizon(\timehorizon) + \E^\mQ*{
		\sum_{t=\riskhorizon+1}^{\timehorizon} \beta_{\riskhorizon}(t) \cfleak{t}
		| \cF_{\riskhorizon}
	}.
\end{equation}
Hence, the original direct and indirect methods no longer converge to the same value. Their expectations differ by the risk-neutral present value of all future leakage cash flows.

Equivalence can be restored by grouping all non-policyholder cash flows on the original shareholder side of the valuation. We achieve this by defining an adjusted indirect present value which explicitly accounts for these leakage costs:
\begin{equation}
	\label{eq:ind_leakage}
	\PVindleak_\riskhorizon(\timehorizon) \coloneqq \sum_{t=\riskhorizon+1}^{\timehorizon} \beta_{\riskhorizon}(t)(\shcf{t} + \cfleak{t}) + \beta_{\riskhorizon}(\timehorizon)F_T^+.
\end{equation}
The corresponding adjusted indirect available capital, $\ACindleak_\riskhorizon(\timehorizon) = \mathbb{E}^\mQ[\PVindleak_\riskhorizon(\timehorizon) \mid \cF_{\riskhorizon}]$, is now, by construction, equal to the direct available capital $\ACdir_\riskhorizon(\timehorizon)$.

Ultimately, the theoretical framework presented in this paper remains intact by simply substituting $\cfsh{t}$ by $\cfsh{t} + \cfleak{t}$. While the downside is that the resulting leak-including indirect estimator no longer perfectly matches the idealized economical derivation presented in \Cref{sec:direct_and_indirect_method}, incorporating these frictions provides a strictly more realistic and necessary basis for practical implementation in insurance markets.

\section{Numerical study}\label{sec:numerical_study}

We study the performance of direct, indirect, and mixed estimator for the available capital on two different models.
We denote the first model, presented in \cite{bauer_RiskneutralValuationParticipating_2006} and further used in \cite{zaglauer_RiskneutralValuationParticipating_2008} and \cite{bauer_CalculationSolvencyCapital_2012a}, as Bauer's model. It is a simple base line model for the asset-liability management of a life insurer with two main cases: the MUST and the IS case, where the first implements what life insurers have to share with their policyholders at a minimum, and the second has a more realistic profit participation for the policyholders.
Another ALM model we test our methods on is openIRM. It is a more realistic representation of an internal risk model under the Solvency II regulatory framework.

\subsection{Bauer's model}

\cite{bauer_RiskneutralValuationParticipating_2006} present a simplified model of a German life insurer with two variants for their policyholder cash flows, which was extended by \cite{zaglauer_RiskneutralValuationParticipating_2008} to include a stochastic interest rate. The first MUST case represents an insurer giving the regulatory minimum of participation to policyholders. The second IS case represents a more realistic version where the amount of participation the policyholders get depends on past market performance and a target reserve rate. We use the same notation as in the abstract ALM model for similar concepts.

To model the asset process in Bauer's model, we follow \cite{bauer_CalculationSolvencyCapital_2012a} and implement a generalized Black-Scholes model with stochastic interest rate $r_t$ described by the Vasicek model, i.e., under the risk-neutral measure $\mQ$ we have
\begin{align}
\begin{split}
	d\assets{t}{}
	&=
	\assets{t}{} \left(
		r_t dt
		+ \shortrateAssetCorr \assetVola d W^r_t
		+ \sqrt{1 - \shortrateAssetCorr^2} \assetVola d W^{\assets{}{}}_t,
	\right)
	\;\; \assets{0}{} > 0, \\
	d r_t
	&=
	\shortrateReversionSpeed (\shortrateMeanLevelRiskNeutral - r_t) dt
	+ \shortrateVola d W^r_t,
	\hspace{2.8cm}\; r_0 > 0,
\end{split}
\end{align}
where $\shortrateMeanLevelRiskNeutral = \shortrateMeanLevel - \marketPriceOfRisk \shortrateVola / \shortrateReversionSpeed$ and $\marketPriceOfRisk$ is the market price of risk. Because we have to integrate cash flows interacting with the asset process, the final implementation of the asset process is given by
\begin{align}
	\assetsBefore{t} = \assetsAfter{t-1}
	\exp \left(
		\int_{t-1}^{t} r_s ds
		- \frac{\assetVola^2}{2}
		+ \shortrateAssetCorr \assetVola (W^r_t - W^r_{t-1})
		+ \sqrt{1 - \shortrateAssetCorr^2} \assetVola (W^A_t - W^A_{t-1})
	\right).
\end{align}
To ensure the needed risk-neutral assumption, we use the exact implementation for Gaussian short-rate models described in \cite[Sec.~3.3, p.~115]{glasserman_MonteCarloMethods_2010}.

We implement an annual updating scheme of the policy reserves and shareholder cash flows, so let $a$ be the length of one year such that the index $t-a$ now refers to the last annual value of the respective process. For all subannual time steps $t$, we set all cash flows (e.g., dividends) to zero and keep the policy reserves $\liabsBefore{t}$ constant. A visualization of a single representative path for the MUST and IS case can be found in \Cref{fig:comparison_must_and_is_case_model}.

\begin{rem}
	Bauer's model is a risk-neutral abstract ALM model.
\end{rem}

\subsubsection{The MUST case}

\begin{figure}[tb]
	\centering
	\includegraphics[width=1\textwidth]{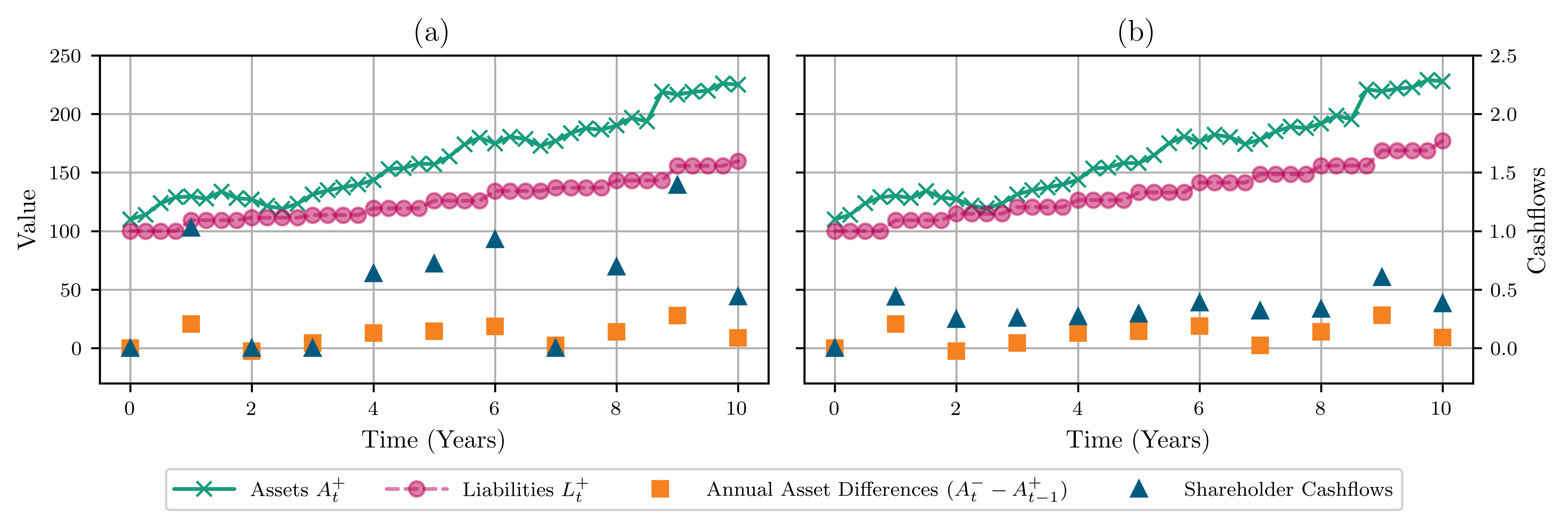}
	\caption{Bauer's model with (a) MUST and (b) IS case for a single representative simulation. Shareholder cash flow is plotted on the right y-axis, with all other quantities corresponding to the left y-axis.}
	\label{fig:comparison_must_and_is_case_model}
\end{figure}

In view of a typical life insurance contract, the policyholders gain a minimum guaranteed interest rate $\guarInt$, hence, $\liabsBefore{t} \geq (1 + \guarInt) \liabsAfter{t-a}$ for all $t$. We assume that at least $\participationRate$ of the earnings on book values have to be credited to the policyholder's accounts and that at least a portion $\earningsFactor$ of any increase in market value has to be identified as earnings on book values. Let
\begin{align}
	\writtenEarnings{t} \coloneqq \participationRate \, \earningsFactor \,
	(\assetsBefore{t} - \assetsAfter{t-a}),
\end{align}
then,
\begin{align}\label{eq:minimum_participation_rate}
	\liabsBefore{t}
	=
	(1 + \guarInt) \liabsAfter{t-a}
	+ \left[
		\writtenEarnings{t}
		- \guarInt \liabsAfter{t-a}
	\right]^+.
\end{align}
The remaining portion is paid out as dividends $\dividends{t}$ to shareholders, i.e.,
\begin{align}
	\dividends{t}
	=
	\begin{cases}
		\earningsFactor (\assetsBefore{t} - \assetsAfter{t-a}) - \writtenEarnings{t}
		&\text{ if }
		\writtenEarnings{t}
		> \guarInt \liabsAfter{t-a}, \\
		\earningsFactor (\assetsBefore{t} - \assetsAfter{t-a})
		- \guarInt \liabsAfter{t-a}
		&\text{ if }
		\writtenEarnings{t}
		\leq \guarInt \liabsAfter{t-a}
		\leq \earningsFactor (\assetsBefore{t} - \assetsAfter{t-a}), \\
		0 &\text{ else.}
	\end{cases}
\end{align}

\subsubsection{The IS case}

Here, a target rate of interest $\targetRate > \guarInt$ is credited to the policy reserves, as long as the so called reserve quota $x_t = (\assetsBefore{t} - \liabsBefore{t}) / \liabsBefore{t}$ stays within a given range $[a, b]$. If the reserve quota leaves this range, the surplus is adjusted. Let $\surplusPortion$ be the portion of any surplus credited to the policyholders. Let $\xi \in \{ \targetRate, \guarInt \}$ and
\begin{align}
	\liabsBefore{t}(\xi) \coloneqq (1 + \xi) \liabsAfter{t},
\end{align}
then the resulting reserve quota would be
\begin{align}
	x_t(\xi) \coloneqq \frac{\assetsBefore{t} - \liabsBefore{t}(\xi)}{\liabsBefore{t}(\xi)},
\end{align}
if the rate of $\xi$ is given to the policyholders.
Now, the amount of policy reserves and dividends are decided based on the following decision rule:

\paragraph{Case 1) $x_t(\targetRate) \in [a, b]$:}
Target rate results in acceptable policy reserves, so
\begin{align}
	\liabsBefore{t} = \liabsBefore{t}(\targetRate),
	\hspace{1cm}
	\dividends{t} = \surplusPortion (\targetRate - \guarInt) \liabsAfter{t-a}.
\end{align}

\paragraph{Case 2) $x_t(\targetRate) < a < x_t(\guarInt)$:}
Company credits the amount such that reserve quota is at lowest acceptable value a. Then,
\begin{align}
\begin{split}
	\liabsBefore{t} &= \liabsBefore{t}(\guarInt)
	+ \frac{1}{1 + a + \surplusPortion} \left[
		\assetsBefore{t} - (1 + \guarInt) (1 + a) \liabsAfter{t-a}
	\right], \\
	\dividends{t} &= \frac{\surplusPortion}{1 + a + \surplusPortion} \left[
		\assetsBefore{t} - (1 + \guarInt) (1 + a) \liabsAfter{t-a}
	\right].
\end{split}
\end{align}

\paragraph{Case 3) $x_t(\guarInt) < a$:}
The resulting reserve level is outside the acceptable range, even when applying the guaranteed interest rate. Therefore, no dividends will be paid,
\begin{align}
	\liabsBefore{t} = \liabsBefore{t}(\guarInt),
	\hspace{1cm}
	\dividends{t} = 0.
\end{align}

\paragraph{Case 4) $x_t(\targetRate) > b$:}
Crediting target rate would exceed acceptable range, so it is capped at the upper limit $b$, i.e.,
\begin{align}
\begin{split}
	\liabsBefore{t} &= \liabsBefore{t}(\guarInt) + \frac{1}{1 + b + \surplusPortion} \left[
		\assetsBefore{t} - (1 + \guarInt) (1 + b) \liabsAfter{t-a}
	\right], \\
	\dividends{t} &= \frac{\surplusPortion}{1 + b + \surplusPortion} \left[
		\assetsBefore{t} - (1 + \guarInt) (1 + b) \liabsAfter{t-a}
	\right].
\end{split}
\end{align}

Finally, the new liabilities have to be at least the minimum participation rate, cf.\ \eqref{eq:minimum_participation_rate}. Hence, if necessary, we overwrite the respective case above and set
\begin{align}
	\liabsBefore{t}
	=
	(1 + \guarInt) \liabsAfter{t-a}
	+ \left[
	\participationRate \, \earningsFactor \,
	(\assetsBefore{t} - \assetsAfter{t-a})
	- \guarInt \liabsAfter{t-a}
	\right]^+,
	\hspace{1cm}
	\dividends{t}
	=
	\surplusPortion \left[
		\participationRate \earningsFactor (\assetsBefore{t} - \assetsAfter{t-a})
		- \guarInt \liabsAfter{t-a}
	\right]^+.
\end{align}

\subsubsection{Optional policyholder cash flows}

We also extended the model to include optional policyholder cash flows, such as constant premium payments and stochastic lapses. However, under our base parameterization, these events were not frequent or material enough to significantly alter the relative variance of the estimators, hence we do not include these extensions for our study. A deeper analysis under more extreme assumptions is left for future research.

\subsection{openIRM}

\cite{wolf_OpenIRMPubliclyAccessible_2025} developed openIRM, a publicly available Internal Risk Model of an artificial life insurer, designed specifically for comparing methods for Solvency Capital Requirement (SCR) estimation.
At its heart lies the available capital estimation which will be our main focus. The model combines an Economic Scenario Generator, which models interest rates and stock dynamics based on the Gaussian 2-Factor (G2++) model, with a cash flow projection model for policies with guaranteed benefits. A key feature of openIRM is its calibration to real market data from 2016 to 2023, allowing it to generate realistic distributions of an insurer's available capital and provide a robust environment for evaluating the performance of our approach.

\begin{rem}
	openIRM is a risk-neutral abstract ALM model.
\end{rem}

\FloatBarrier
\subsection{Comparison of direct and indirect estimator}

We now compare the performance of direct and indirect estimator for the presented models in various parameter settings.
In \Cref{fig:direct_minus_indirect_shortrate_liabs}, the difference of variance of direct and indirect estimator based on $10\,000$ simulations is depicted for guaranteed interest rate $\guarInt$ and initial short rate $r_0$ in the a) MUST and b) IS case. These figures highlight the non-trivial relationship of direct and indirect method depending on the chosen model and parameters.
Importantly, we see in b) that the curve of parameters where both methods have the same variance seems to be non-linear.
In practice, the guaranteed interest rate $\guarInt$ is usually lower than the initial short rate $r_0$. For all of these cases, the indirect method would be the preferred method in the more realistic IS case.
In \Cref{fig:direct_minus_indirect_y_guaranteed_interest_rate_annual}, we see the same difference but for the factor giving the participation of policyholders $\earningsFactor \participationRate$ and again $\guarInt$. Here, the border may even be non-monotonic, as shown in (b).

These figures reinforce the dynamic observed in the stylized setting of \Cref{rem:edge_cases_of_estimators_variance}: the direct estimator exhibits smaller variance when liabilities are less dependent on assets. This occurs, for instance, when the guaranteed interest rate $\guarInt$ exceeds the short rate $\shortrate_t(\omega)$ in most scenarios $\omega$, when policyholder participation in market earnings $\earningsFactor \participationRate$ is low, or when the insurance design allocates more of the surplus to the policyholders like in the IS case.
These findings support the hypothesis that weaker asset-liability coupling favors the direct method, whereas stronger coupling favors the indirect method. However, the non-monotonic boundary in \Cref{fig:direct_minus_indirect_y_guaranteed_interest_rate_annual}b suggests that the amount of coupling is not always immediately apparent from the chosen parameterization.

\begin{figure}[t]
	\begin{subfigure}{0.5\linewidth}
		\centering
		\caption{}
		\includegraphics[width=1\textwidth]{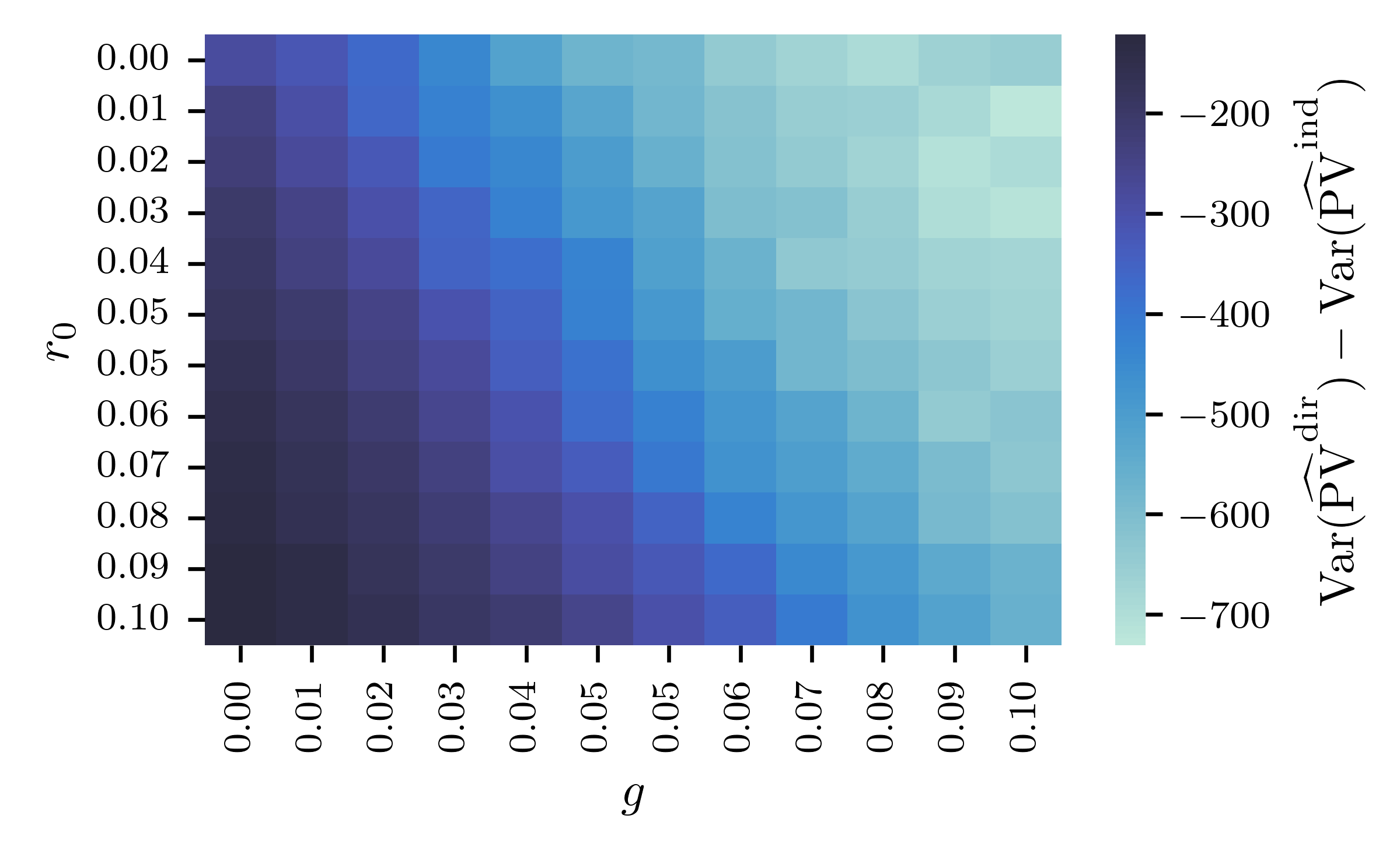}
	\end{subfigure}
	\begin{subfigure}{0.5\linewidth}
		\centering
		\caption{}
		\includegraphics[width=1\textwidth]{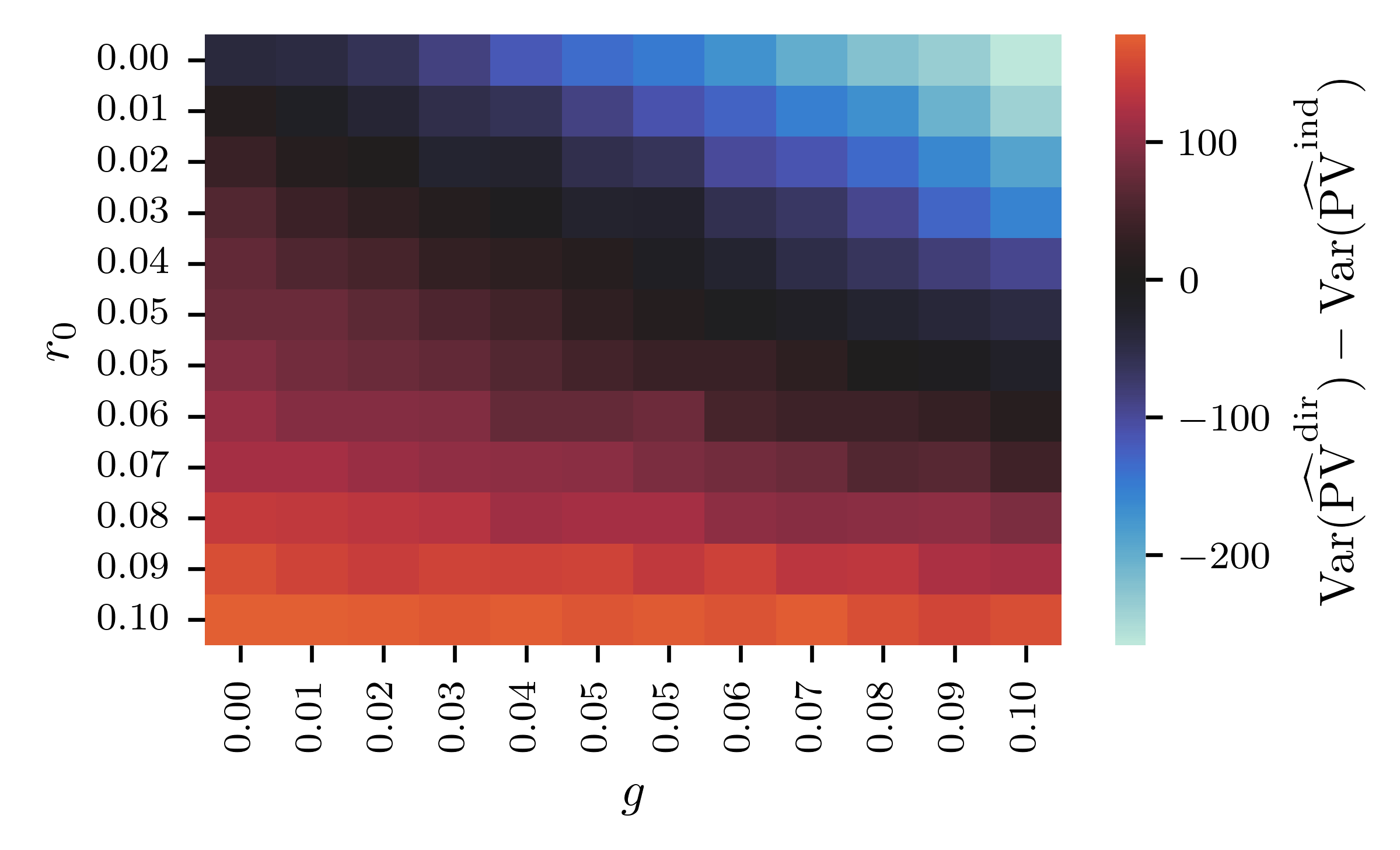}
	\end{subfigure}
	\caption{Difference of variance of direct and indirect estimator for the available capital based on $10\,000$ realizations for varying guaranteed interest rate $\guarInt$ and initial short rate $r_0$ for (a) MUST and (b) IS case.}
	\label{fig:direct_minus_indirect_shortrate_liabs}
\end{figure}

\begin{figure}[t]
	\begin{subfigure}{0.5\linewidth}
		\centering
		\caption{}
		\includegraphics[width=1\textwidth]{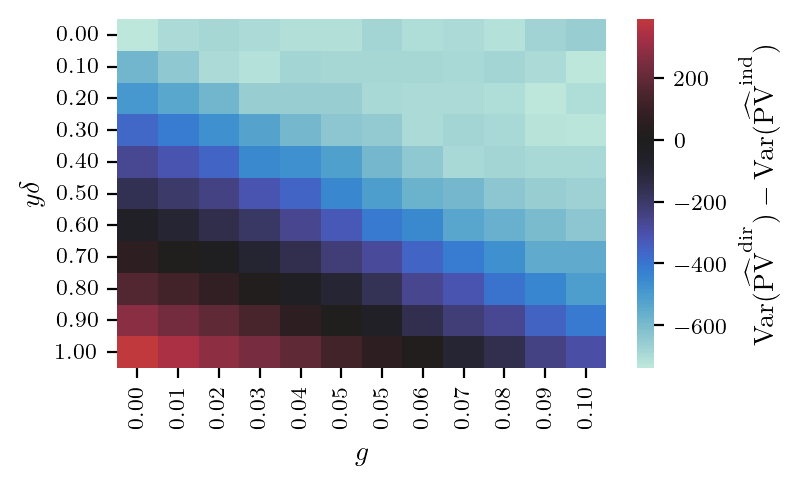}
	\end{subfigure}
	\begin{subfigure}{0.5\linewidth}
		\centering
		\caption{}
		\includegraphics[width=1\textwidth]{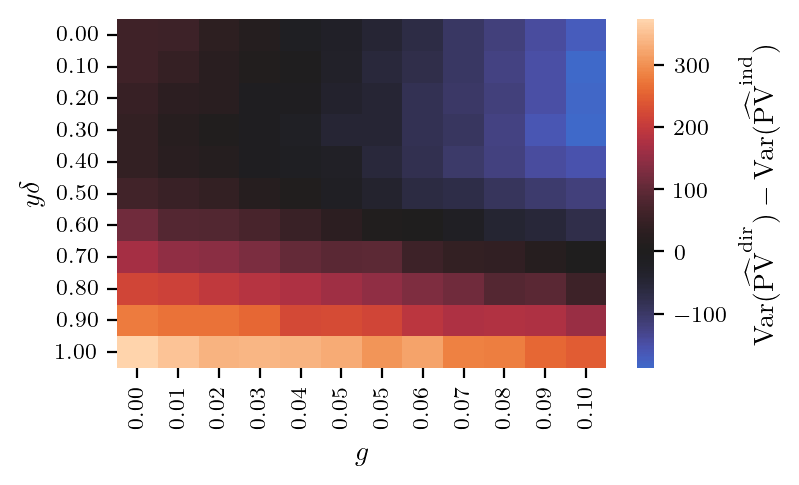}
	\end{subfigure}
	\caption{Difference of variance of direct and indirect estimator for the available capital based on $10\,000$ realizations for varying guaranteed interest rate $\guarInt$ and earnings factor times participation rate $\earningsFactor \participationRate$ for (a) MUST and (b) IS case.}
	\label{fig:direct_minus_indirect_y_guaranteed_interest_rate_annual}
\end{figure}

In \Cref{fig:multiplot_vary_param_MUST} and \Cref{fig:multiplot_vary_param_IS}, we vary a single parameter in Bauer's model for the MUST and IS case, respectively. For each parameter set, we estimate the available capital $10\;000$ times with the results for the direct method shown as the left box plot in teal, and the results for the indirect method shown on the right in red. Their respective mean is shown in front of the boxes as a diamond.
Figures~\ref{fig:multiplot_vary_param_MUST} to \ref{fig:openIRM_per_parameter_boxen_direct_vs_indirect_100_outliers_removed} display the empirical distributions using letter-value plots (boxenplots). Unlike standard boxplots, the varying widths represent successive quantiles extending into the tails, allowing for a precise visual comparison of skewness and heavy tails, which are obscured if one only reports the variance.

While the preceding figures already illustrated the difference in estimator variances, we deliberately present the full empirical distributions here. Because the underlying asset-liability cash flows frequently produce skewed distributions and heavy tails (see, for example, \Cref{fig:multiplot_vary_param_MUST}e), evaluating estimator performance based solely on variance is scientifically insufficient.

Analyzing these full distributions allows for multiple observations. The choice of parameter ranges heavily influences the resulting available capital as well as the estimators' variance for both methods, see \Cref{fig:multiplot_vary_param_MUST}a. Because of this, we chose the ranges to be realistic while still allowing us valuable insights into possible dependencies or stark variations.
The dispersion of the distributions varies heavily depending on the underlying parameters, see \Cref{fig:multiplot_vary_param_MUST}e. Although the direct method has a smaller variance in the MUST case for most settings, \Cref{fig:multiplot_vary_param_MUST}f indicates that this does not hold everywhere, confirming our previous analysis in \Cref{rem:edge_cases_of_estimators_variance}.
We can see that the proven equality of the expected value always holds. Also, the distributions clearly differ in general shape and skewness, further validating the need to study them beyond their pure variance.

For the more realistic IS case, the general behavior is similar, but one crucial difference can be seen. For this case, the indirect method has a smaller variance for most settings, although sometimes, see \Cref{fig:multiplot_vary_param_IS}d and~h, the direct method can still be narrower.

\begin{figure}[tb]
	\centering
	\includegraphics[width=1\textwidth]{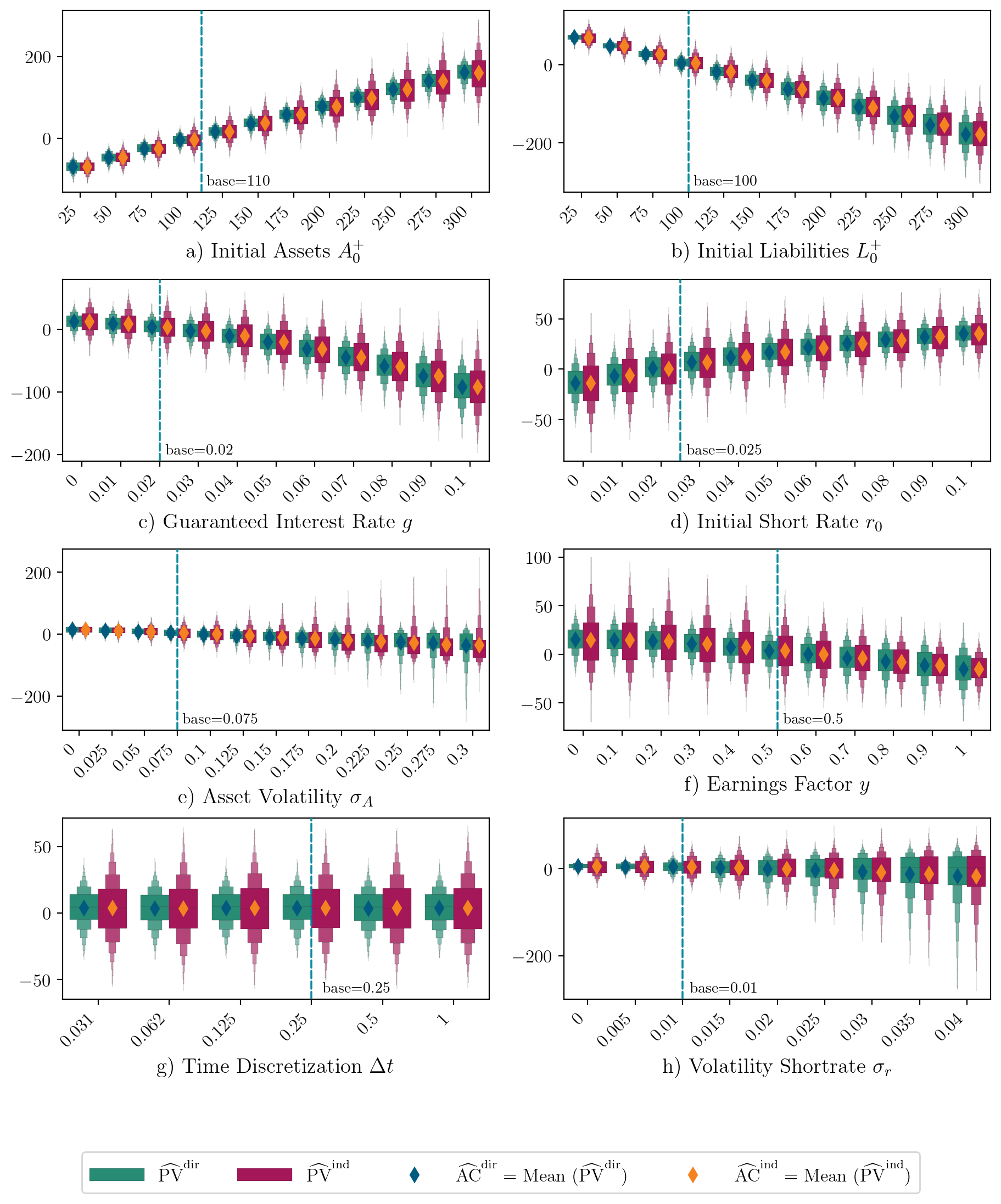}
	\caption{Distribution of $\ACdirest$ and $\ACindest$ for varying parameters in Bauer's model, MUST case. The $100$ outermost values on each side were removed for better visual clarity, see \Cref{fig:appendix_multiplot_vary_param_MUST} for the version with no removed values.}
	\label{fig:multiplot_vary_param_MUST}
\end{figure}

\begin{figure}[tb]
	\centering
	\includegraphics[width=1\textwidth]{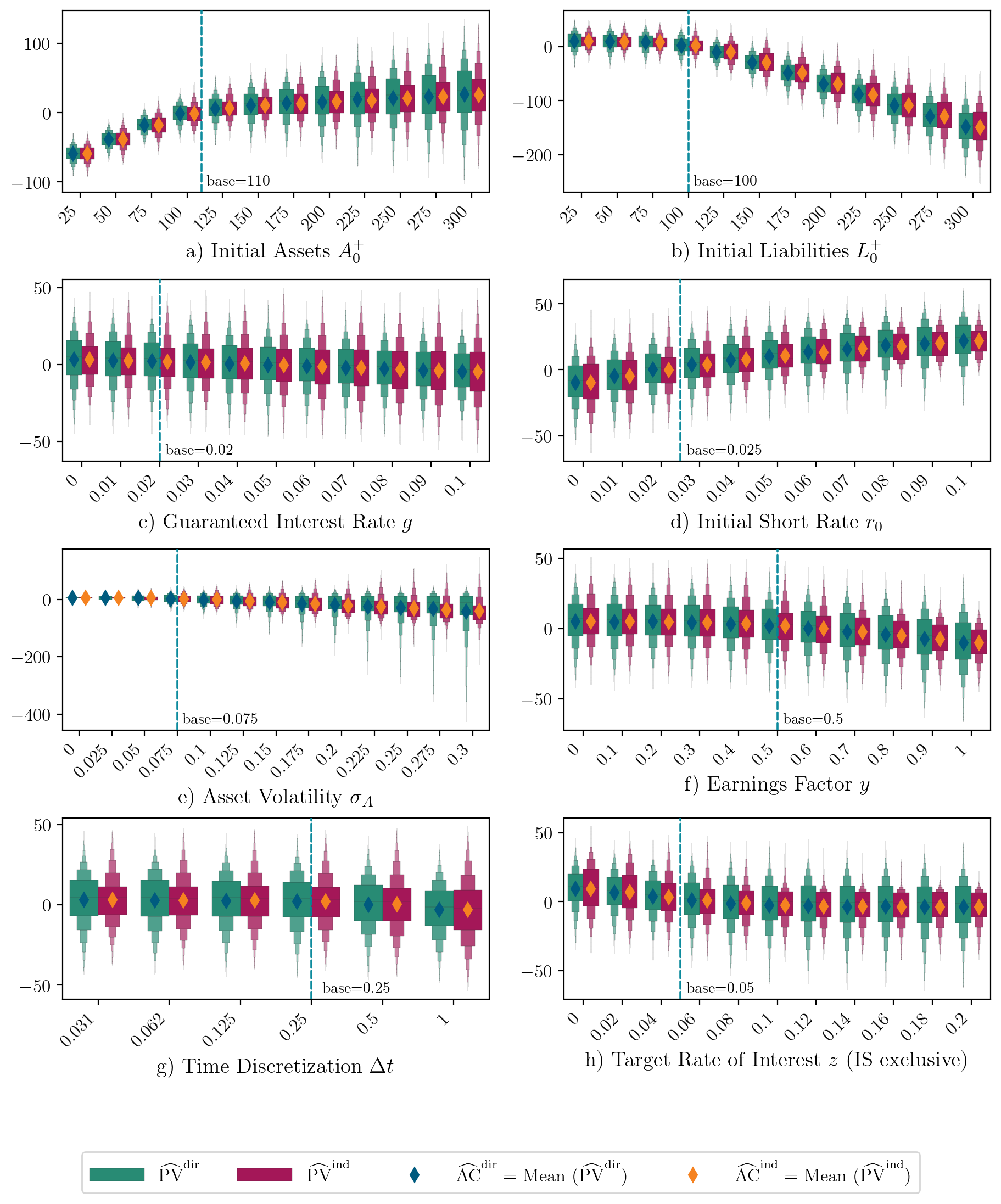}
	\caption{Distribution of $\ACdirest$ and $\ACindest$ for varying parameters in Bauer's model, IS case. The $100$ outermost values on each side were removed for better visual clarity, see \Cref{fig:appendix_multiplot_vary_param_IS} for the version with no removed values.}
	\label{fig:multiplot_vary_param_IS}
\end{figure}

We present the same comparison of direct and indirect distribution for the base setting and varying a single parameter for openIRM by \cite{wolf_OpenIRMPubliclyAccessible_2025} in \Cref{fig:openIRM_per_parameter_boxen_direct_vs_indirect_100_outliers_removed}.
Starting at the risk horizon of $1$ year, we run the inner simulation of openIRM used to estimate the market-value of liabilities and consequently the available capital. We vary all used risk factors and the initial values (at year $1$) of the underlying stochastic processes of the capital market using a Gaussian 2-Factor (G2++) model. That is, the short rate is given by $r_t = x_t + y_t + \psi_t$, $\psi_t$ being deterministic, and $S_t$ is the stock process. For more details, see the original paper.

We observe, that generally the indirect method results in a much narrower distribution. This only changes in settings where we have a low interest rate environment, cf.\ \ref{fig:openIRM_per_parameter_boxen_direct_vs_indirect_100_outliers_removed}a or~\ref{fig:openIRM_per_parameter_boxen_direct_vs_indirect_100_outliers_removed}d, where the yield curve PC1 roughly translates to a change in overall level of the interest rate curve. We can also observe that the available capital of direct and indirect method, given as the mean of the estimators, is approximately the same for each individual setting, confirming \Cref{thm:equality_dir_ind} in more complicated models.

\begin{figure}[tb]
	\centering
	\includegraphics[width=1\textwidth]{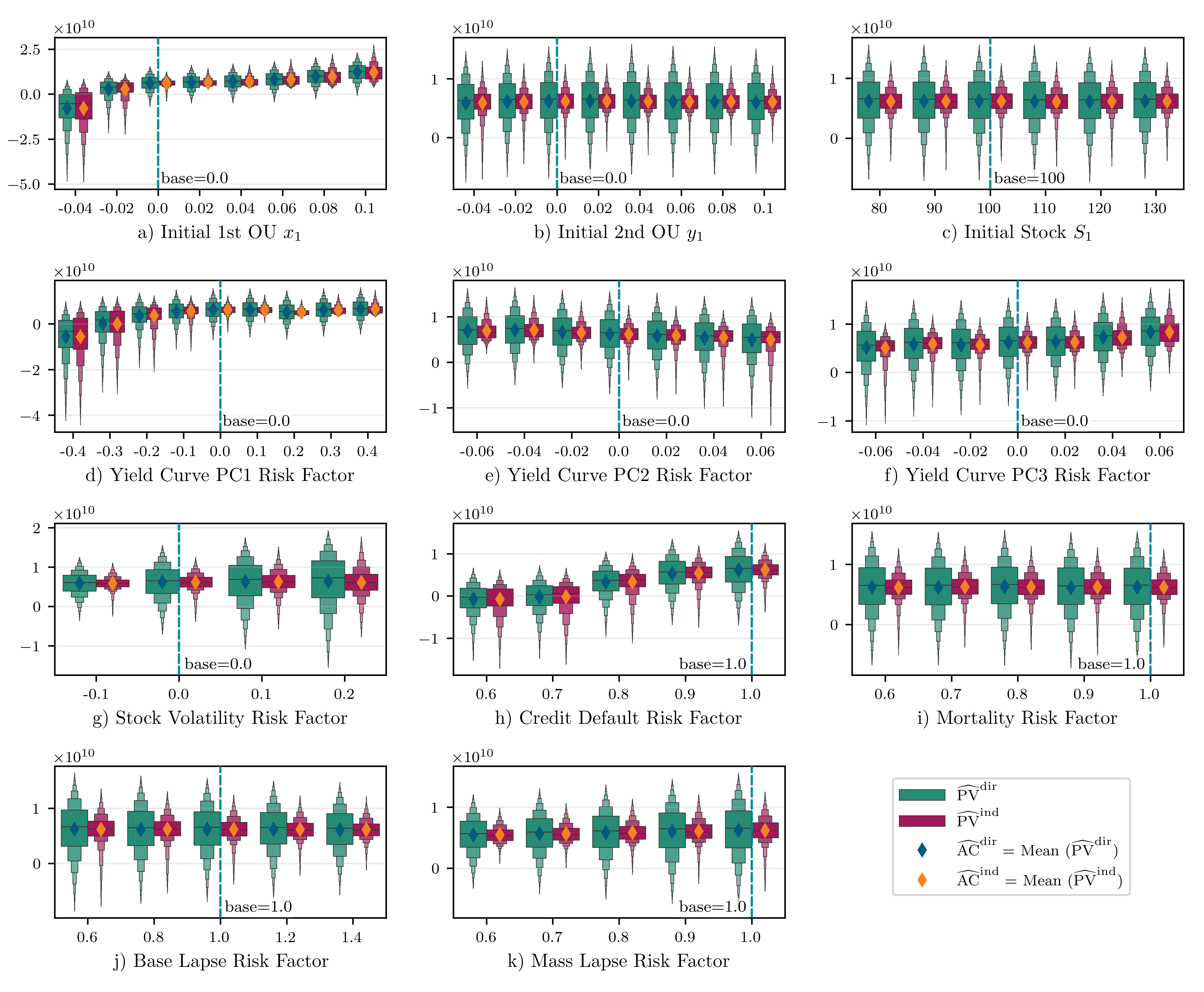}
	\caption{Distribution of $\ACdirest$ and $\ACindest$ for varying parameters in openIRM. The $100$ outermost values on each side were removed for better visual clarity, see \Cref{fig:appendix_openIRM_per_parameter_boxen_direct_vs_indirect_0_outliers_removed} for the version with no removed values.}
	\label{fig:openIRM_per_parameter_boxen_direct_vs_indirect_100_outliers_removed}
\end{figure}

\FloatBarrier
\subsection{Control variates}

\subsubsection*{Studied control variates}

For the following analysis, we look at two control variates in particular, the crude control variate and the mixed control variate.

\begin{defi}\label{def:crude_control_variate_estimator}
	We define the \emph{crude present value control variate estimator} as the combination of the direct and indirect estimator for the available capital, i.e.,
	\begin{align}
		\PVcv_\riskhorizon(b) = \PVdir_\riskhorizon - b \left( \PVdir_\riskhorizon - \PVind_\riskhorizon \right).
	\end{align}
\end{defi}

\begin{defi}\label{def:mixed_control_variate_estimator}
	We define the \emph{mixed present value control variate estimators} for our study as the multi-control estimator combining direct and indirect method as well as mixed present values with all atoms $\{t\}, t \in \timeInMixed_{\mathrm{all}} = \{1, 2, \dots, \timehorizon \}$, i.e.,
	\begin{align}
		\PVmix_\riskhorizon(b)
		\coloneqq
		\PVdir_\riskhorizon
		- b_0 \left( \PVdir_\riskhorizon - \PVind_\riskhorizon \right)
		- \sum_{t \in \timeInMixed_{\mathrm{all}}} b_t \left( \PVdir_\riskhorizon - \PVmix_\riskhorizon(\timehorizon, \{t\}) \right).
	\end{align}
\end{defi}

\begin{rem}
	Combining previous results,
	\begin{align}
		\ACdir_\riskhorizon
		= \E^{\mQ}*{\PVcv_\riskhorizon(b) | \cF_\riskhorizon}
		= \E^{\mQ}*{\PVmix_\riskhorizon(b) | \cF_\riskhorizon}.
	\end{align}
\end{rem}

In the following experimental results, we have again $\riskhorizon = 0$ for Bauer's model and $\riskhorizon = 1$ for openIRM, due to openIRM including the outer simulation over one year of the Solvency II scheme.

\subsubsection*{Computational costs}

While control variates reduce estimator variance, the required computational overhead may offset these gains, which we will now discuss. We establish the direct estimation method as our computational baseline.

The direct estimator requires all policyholder cash flows and the final free reserve, meaning all relevant balance-sheet positions are already computed during the stochastic simulation. As a result, the calculation of quantities for both the crude and mixed control variate estimators introduces no additional computational cost, though it does require a slight increase in memory to store them.

For a crude control variate with a fixed coefficient $b$, the computational cost is negligibly higher than that of the direct estimator. However, as described in \Cref{sec:control_variates}, estimating the optimal coefficient $b^*$ via $\hat{b}$ introduces additional computational load that depends on the chosen option.
We use the naive approach where we employ the entire data set for estimating $\hat{b}$ as well as the control variate estimator. This introduces a negligible bias scaling in $O(1/N)$ due to them not being independent (cf.\ \cite{glasserman_MonteCarloMethods_2010}).
In our experimental results, we were not able to observe a noticeable bias for small $N$, which we account to the standard error scaling in $O(1/\sqrt{N})$.
Alternatively, one can avoid this bias by using a fixed subset of the observed data to estimate $\hat{b}$ and then use this fixed coefficient to compute the control variate estimator for the remaining observations. This would also keep the computational overhead needed for $\hat{b}$ fixed, but of course would reduce the number $N$ of usable data points for the available capital estimation.

For the mixed control variate, estimating the optimal coefficients $b^*$ and combining the estimators can introduce a noticeable overhead. This overhead scales with the number of time steps, which influences the number of possible mixed estimators. However, we believe that in practice this additional cost is insignificant relative to the primary cost of the stochastic simulation of the balance sheet for a typical life insurer.

\subsubsection*{Experimental Results}

For visual clarity, we first compare the direct, indirect, and crude control variate estimators. We omit $\PVmix(b)$ from the convergence and distribution plots as it behaves similarly to the crude estimator and offers limited additional insight.

\paragraph{Convergence plot.}
In \Cref{fig:estimator_comparison_1x2_with_CI_MUST}, direct, indirect, and the crude control variate estimator are compared for a varying amount of observations $N$. We can see the estimators in \Cref{fig:estimator_comparison_1x2_with_CI_MUST}a, and their approximate $95\%$ confidence intervals (CI) can be seen in \Cref{fig:estimator_comparison_1x2_with_CI_MUST}b.
The same is shown in \Cref{fig:estimator_comparison_1x2_with_CI_IS} for the more realistic IS case of Bauer's model. For both models, the parameterization is fixed to our base case again.
Note that, the crude control variate estimator estimates the coefficient $b$ only with the available subset of $N$ observations.

We observe that for the MUST case, the improvement given by the crude CV estimator is small compared to the direct estimator. For the more realistic IS case, the improvement by the crude CV is clearly visible. The estimator seemingly converges much faster to the true value, and its confidence interval is narrower as well.

\begin{figure}[tb]
	\centering
	\includegraphics[width=1\textwidth]{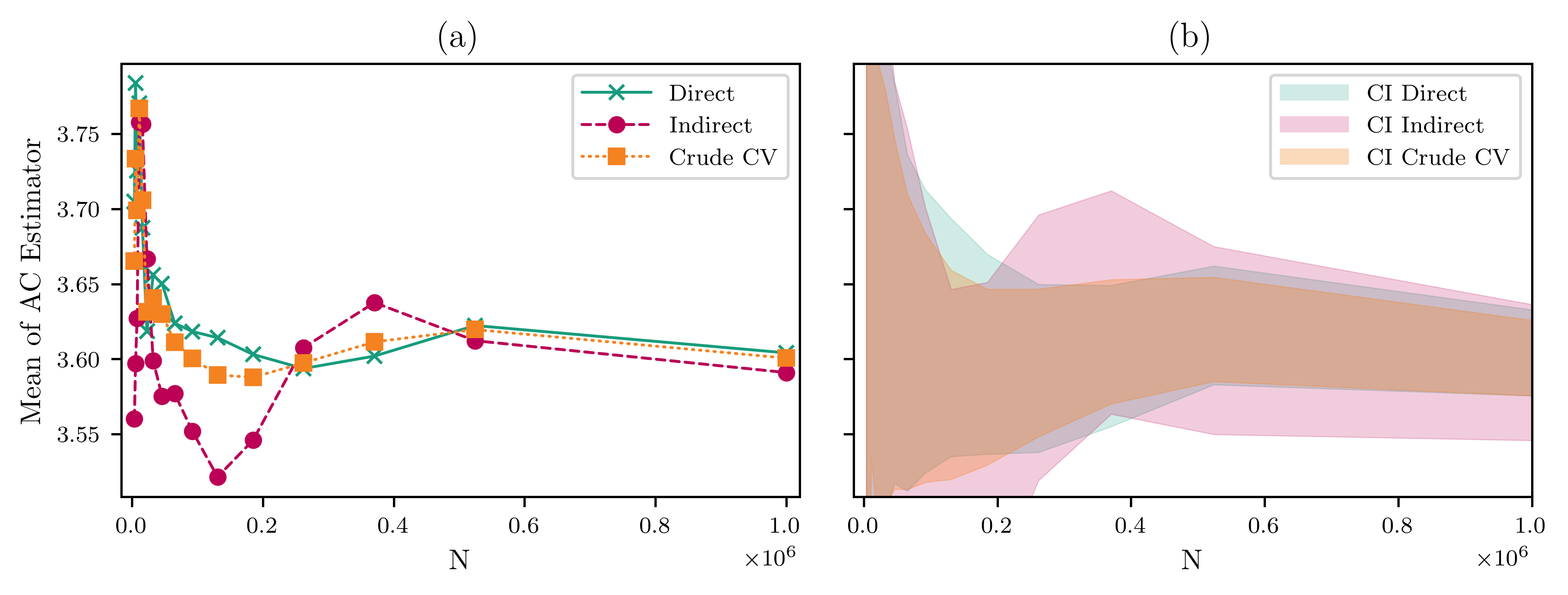}
	\caption{Comparison of direct, indirect, and crude control variate estimator for Bauer's model MUST case for varying number of observations $N$. (a) The estimated value given by the mean. (b) Approximate $95\%$ confidence intervals for estimators.}
	\label{fig:estimator_comparison_1x2_with_CI_MUST}
\end{figure}

\begin{figure}[tb]
	\centering
	\includegraphics[width=1\textwidth]{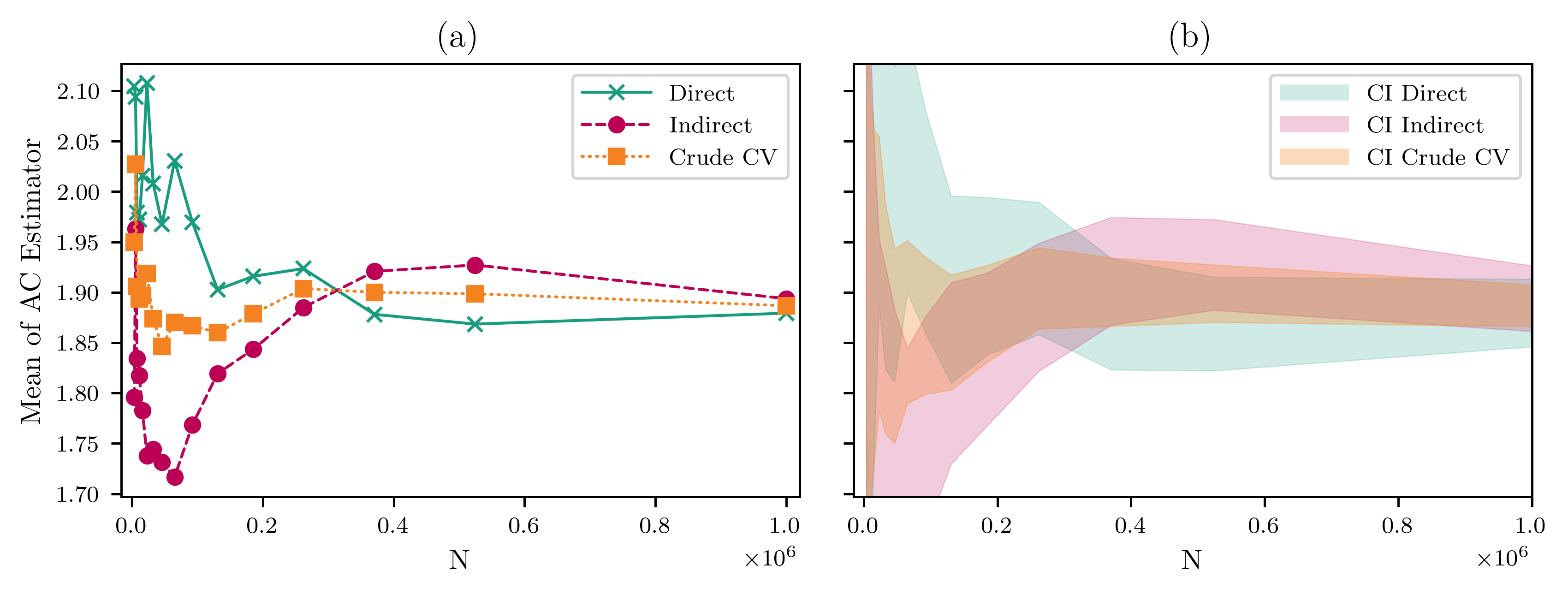}
	\caption{Comparison of direct, indirect, and crude control variate estimator for Bauer's model IS case for varying number of observations $N$. The crude control variate estimator estimates the coefficient $b$ only with the available subset of $N$ observations. (a) The estimated value given by the mean. (b) Approximate $95\%$ confidence intervals for estimators.}
	\label{fig:estimator_comparison_1x2_with_CI_IS}
\end{figure}

\paragraph{Final estimates distributions.}
In \Cref{fig:kde_per_run_estimator_distributions}, we compare the same estimator's distribution in Bauer's model for both (a) MUST and (b) IS case. For each single estimated value, we take the mean over $1\,000$ realizations of the respective random variable. We repeat this $1\,000$ times. The resulting estimator distribution is approximated using the Kernel Density Estimation (KDE) provided by \cite{waskom_SeabornkdeplotSeabornStatistical_2024}, which smooths the observations with a Gaussian kernel to produce a continuous density estimate. We use the base setting for our parameters.

We observe that for the MUST case, the direct estimator's distribution is narrower than the indirect estimator's distribution. The crude control variate estimator only has a slightly more narrow distribution than the direct estimator. This is consistent with our previous findings for the MUST case. For the more realistic IS case, direct and indirect estimator have a similar distribution, but their combination presented in the crude control variate estimator has a clearly narrower shape. Hence, using the crude estimator in the IS case would results in a faster convergence.

\begin{figure}[tb]
	\begin{subfigure}{0.5\linewidth}
		\centering
		\includegraphics[width=1\textwidth]{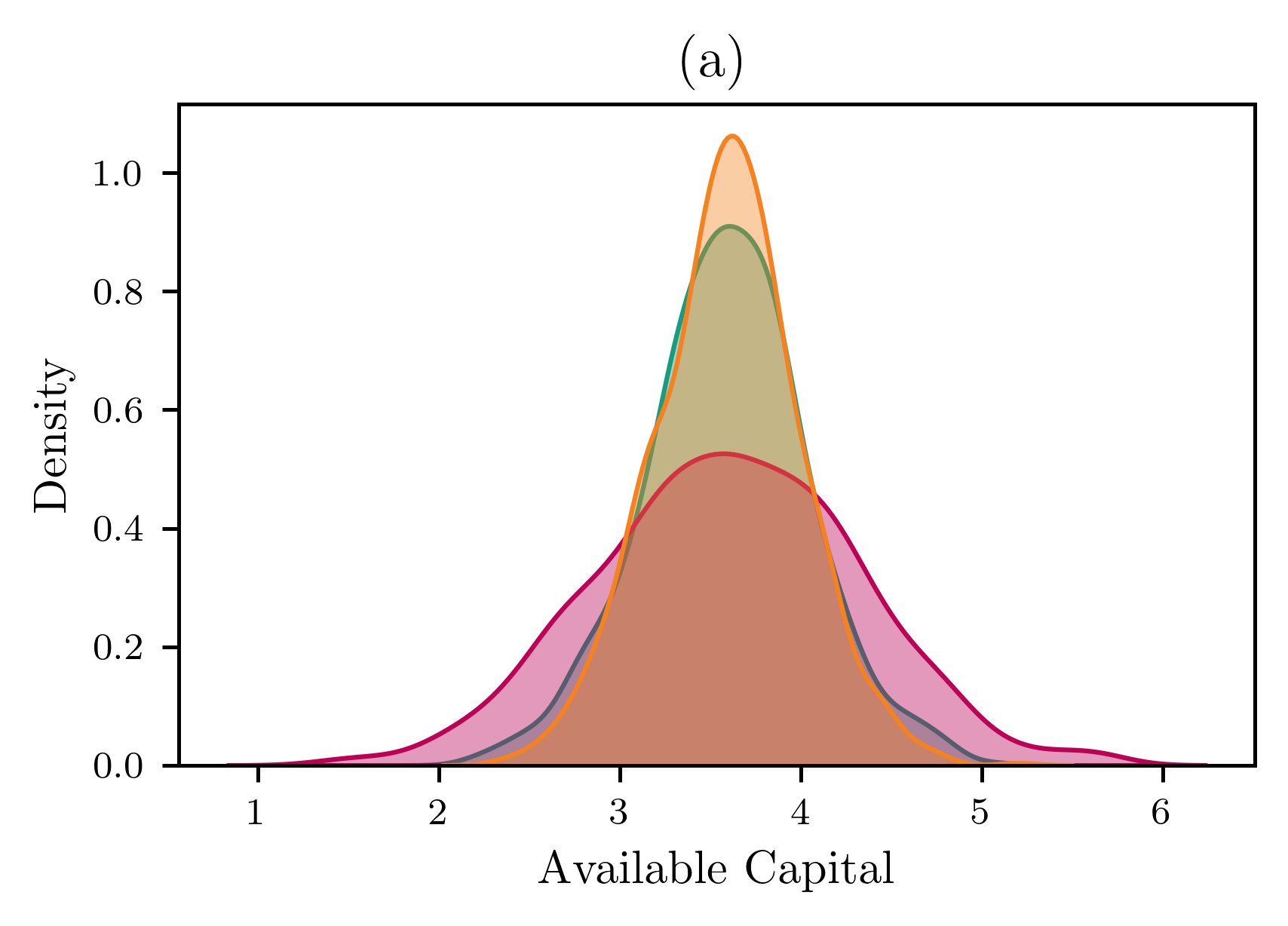}
	\end{subfigure}
	\begin{subfigure}{0.48\linewidth}
		\centering
		\includegraphics[width=1\textwidth]{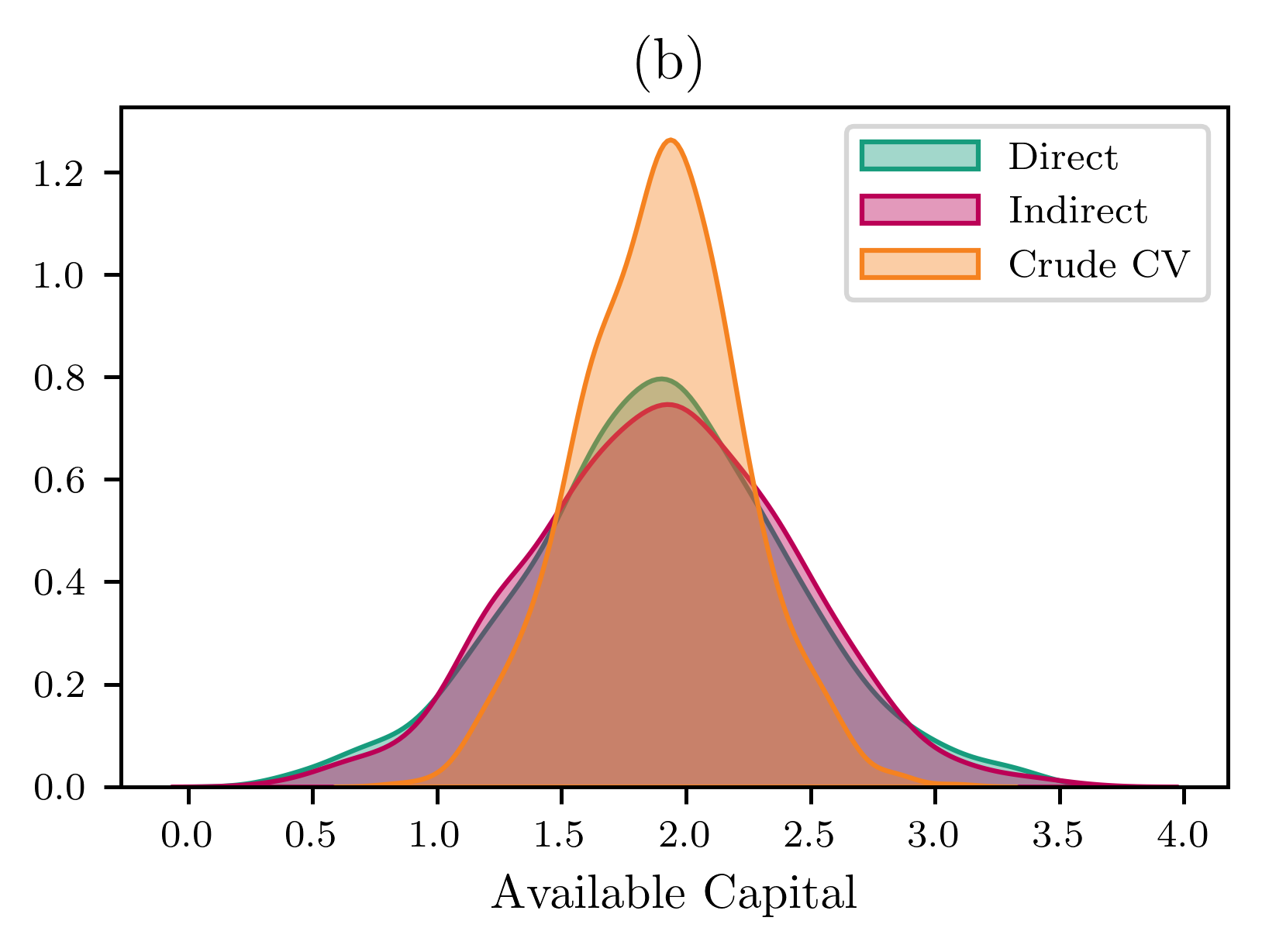}
	\end{subfigure}
	\caption{Comparison of direct, indirect, and crude control variate estimator's distribution (approximated with a Kernel Density Estimation) in Bauer's model (a) MUST and (b) IS case for $1\,000$ different estimations using $1\,000$ simulations each.}
	\label{fig:kde_per_run_estimator_distributions}
\end{figure}

\paragraph{Variance reduction factor of control variate estimators.}
For Bauer's model, we study the variance reduction factor (VRF) for the base setting with one varying parameter in \Cref{fig:multiplot_VRF_Multi_CV_Mixed_2_MUST} for the MUST case and in \Cref{fig:multiplot_VRF_Multi_CV_Mixed_2_IS} for the IS case.
Recall that, the VRF is given as $1 - R^2$ which simplifies to $1 - \rho^2$ in the 1-dimensional case, where $\rho$ is the correlation of target variable and control (cf.\ \Cref{sec:control_variates}). We estimate $\rho$ via its sample correlation over $10\,000$ simulations.
Note that due to the central limit theorem and the confidence interval of the sample mean estimator, the VRF translates one to one to the reduction in necessary Monte-Carlo samples to reach a specific confidence interval length.

\Cref{fig:multiplot_VRF_Multi_CV_Mixed_2_MUST} explains our previous findings about the crude estimator only slightly improving the simple direct estimator in our base setting. We see that for the parameters used, the correlation of $\scriptstyle \ACdirest$ and $\scriptstyle \ACdirest - \ACindest$ is close to $0$, resulting in a variance reduction factor close to $1$ and, hence, only a slight reduction in variance for the resulting control variate estimator. But for some parameters deviating from the base setting, we can observe that the correlation is much higher, cf.\ \Cref{fig:multiplot_VRF_Multi_CV_Mixed_2_MUST}e or~\ref{fig:multiplot_VRF_Multi_CV_Mixed_2_MUST}f.

In all subfigures we see the mixed estimator at least equal to the reduction of the crude estimator. This is expected as the indirect estimator is also used as a control in the mixed estimator, resulting in the mixed estimator having equal to or more information than the crude one for all settings.

In the more realistic IS case shown in \Cref{fig:multiplot_VRF_Multi_CV_Mixed_2_IS}, we can see a tangible improvement of $\PVmix(b)$ for almost all parameter settings. It ranges from $0\%$ to $20\%$ reduction of the variance compared to the crude estimator $\PVcv(b)$.
It is difficult to say if this improvement is worth the effort of using the mixed estimator: because all quantities are already computed during the simulation, the computational overhead is limited and would not offset the improvements. But the initial implementation of the more complicated mixed estimator might take some effort.
We make one additional observation here: the correlation seems to increase the more the randomness of the market results in randomness in policyholder cash flows. Meaning that if policyholders are awarded a larger share of the market returns (large $y$ or $z$) or have less guarantees (small $g$), the direct and indirect estimator have a higher correlation resulting in a better performing crude control variate estimator.
This observation is consistent with the low correlations seen in the MUST case. Here, most of the policyholder's cash flows are determined by the guaranteed interest rate $g$. Small earnings factor $\earningsFactor=0.5$ and policyholder share $\participationRate=0.9$ in the base setting result in a low participation of market returns. \Cref{fig:multiplot_VRF_Multi_CV_Mixed_2_MUST}f indicates that a higher earnings factor dramatically increases the correlation of direct and indirect estimator.

In both cases, the correlation vanishes as the short-rate volatility $\shortrateVola$ increases (cf. subfigure (g)). We attribute this to the increased volatility of the discount factor $\disc{\riskhorizon}{t}$, which seemingly masks the underlying correlation structure of the balance-sheet processes.

In \Cref{fig:appendix_multiplot_VRF_Multi_CV_Mixed_MUST} and \Cref{fig:appendix_multiplot_VRF_Multi_CV_Mixed_IS}, we also plot single control estimators based on different time subsets $\timeInMixed$: $H_1$ and $H_2$ for the first and second halves of the time steps, and $Q_i$ for the $i$-th quarter.

\begin{figure}[tb]
	\renewcommand{\thesubfigure}{\roman{subfigure}}
	\begin{subfigure}{\textwidth}
		\caption{MUST case}
		\includegraphics[width=1\textwidth]{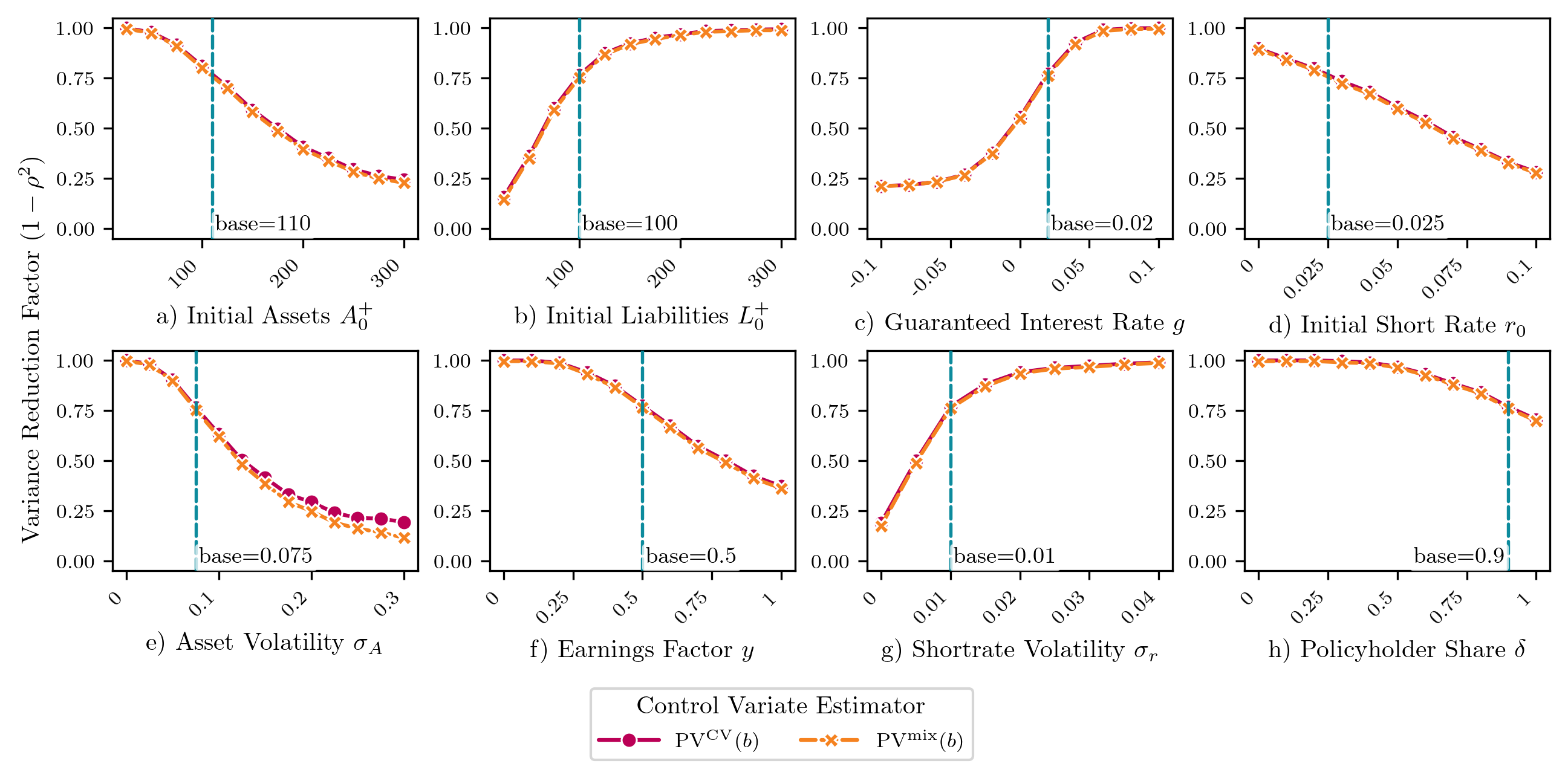}
		\label{fig:multiplot_VRF_Multi_CV_Mixed_2_MUST}
	\end{subfigure}
	\begin{subfigure}{\textwidth}
		\caption{IS case}
		\includegraphics[width=1\textwidth]{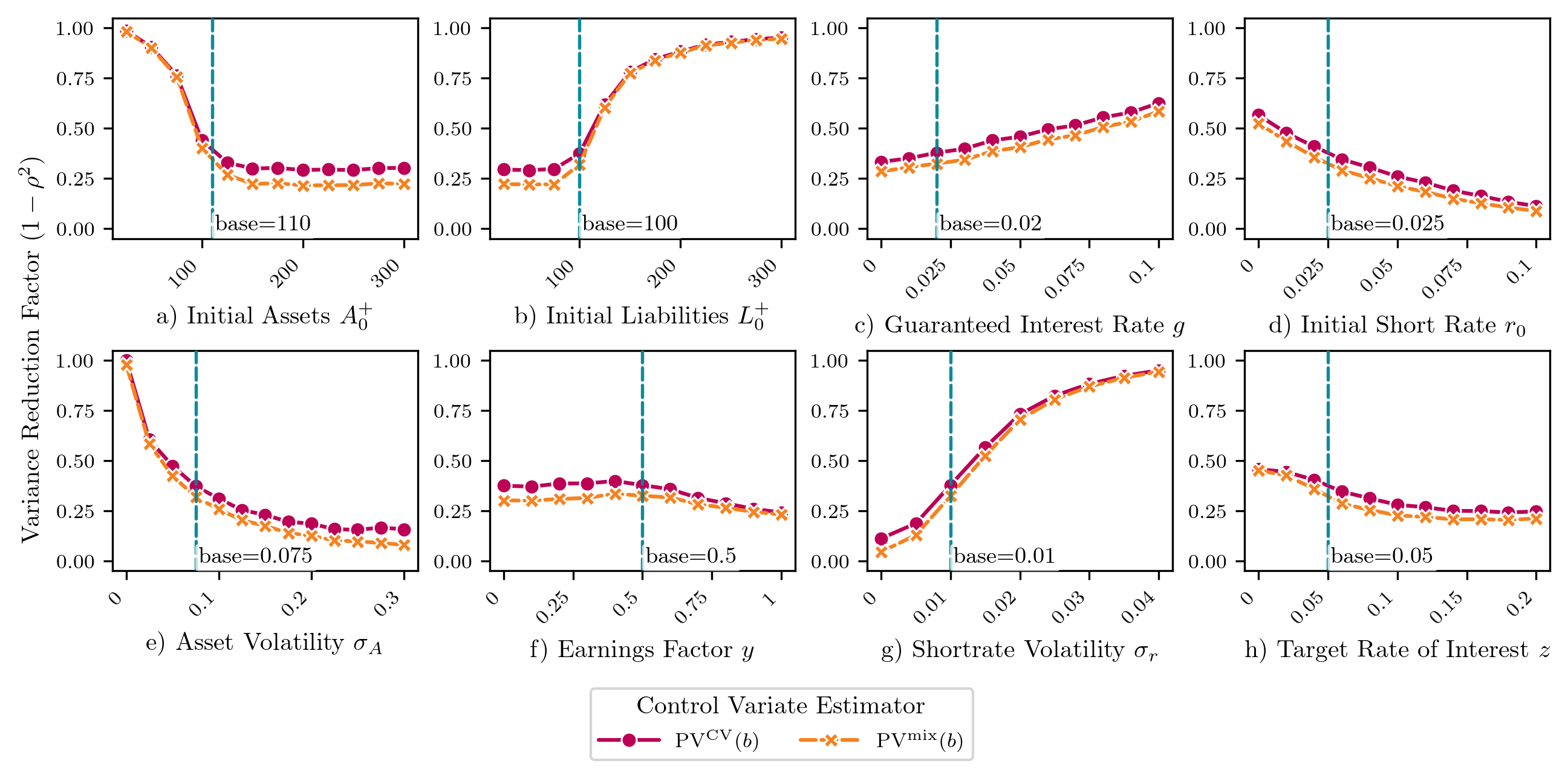}
		\label{fig:multiplot_VRF_Multi_CV_Mixed_2_IS}
	\end{subfigure}
	\caption{Impact of parameter variation on the variance reduction factor for both the crude control variate estimator (with ${\scriptstyle \ACdirest - \ACindest}$ as the control) and the mixed control variate estimator in Bauer's model, (a) MUST, (b) IS case. Each subplot illustrates the effect of varying a single parameter, with the vertical line marking the constant base setting for all other parameters.}
\end{figure}

\FloatBarrier

\begin{figure}[tb]
	\centering
	\includegraphics[width=0.6\textwidth]{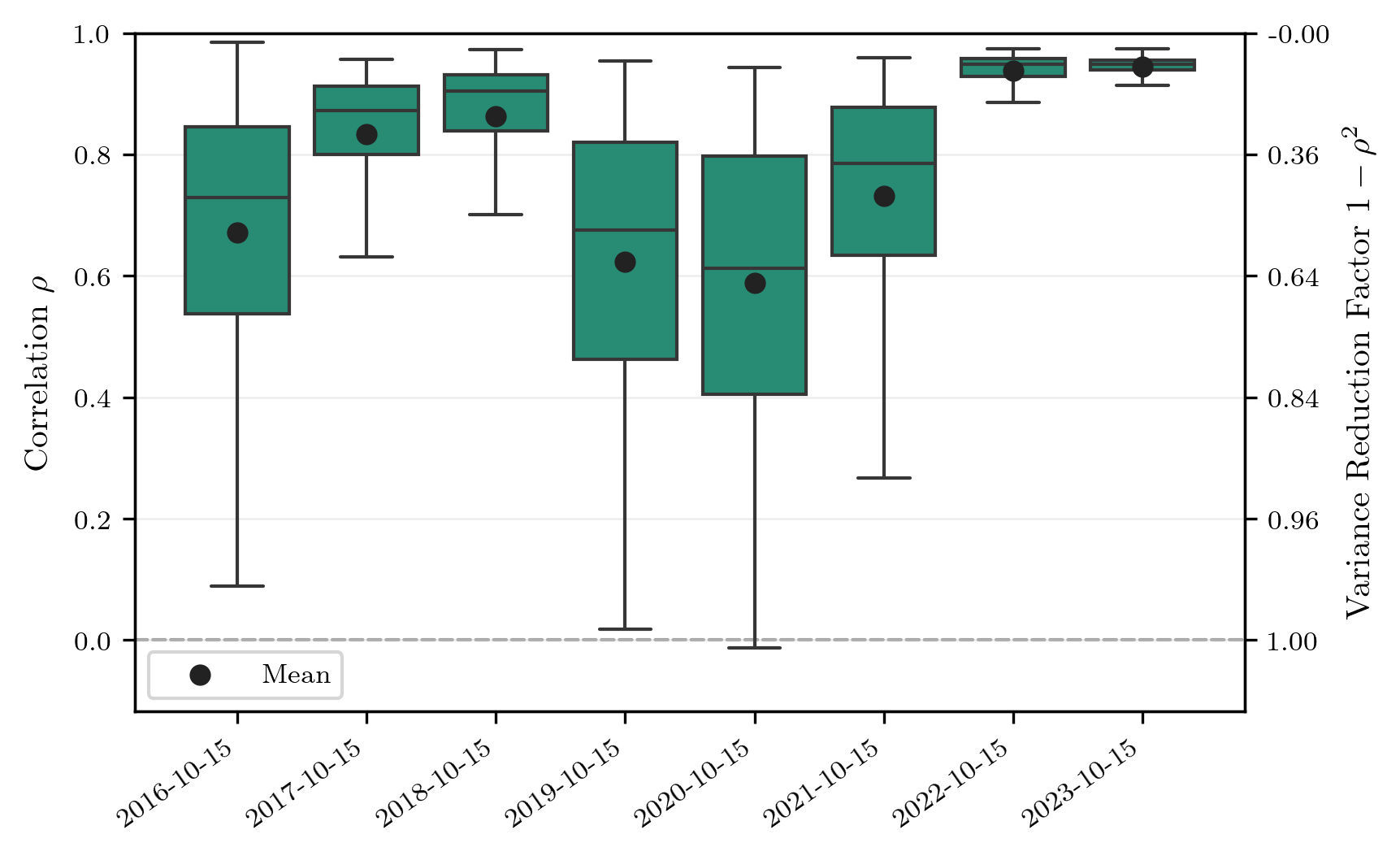}
	\caption{Correlation $\rho$ and variance reduction factor $1 - \rho^2$ of $\ACdirest$ and $\ACdirest - \ACindest$ for $1\,000$ estimations (one for each outer simulation) based on $1\,000$ realizations of the respective present value $\PV$ (one for each inner simulation) for varying years in openIRM.}
	\label{fig:openIRM_dir_ind_cv_correlation_over_years}
\end{figure}

In \Cref{fig:openIRM_dir_ind_cv_correlation_over_years}, we analyze the variance reduction factor of the crude control variate estimator $\PVcv$ in openIRM for varying years. Note that in openIRM, the selected date changes the capital market calibration as well as the current balance sheet, i.e., the initial parameterization of the internal risk model and therefore of our simulation changes substantially.
Then, for each selected date we execute the usual nested Monte-Carlo protocol that is employed for estimating the Solvency Capital Requirement. This is done by simulating $1\,000$ outer paths under the physical measure $\mathbb{P}$ until the risk horizon of $1$ year. Then, a second, inner simulation is done under the risk-neutral measure $\mQ$ to estimate the $\AC$ employing $1\,000$ paths as well. This equal budget allocation is not optimal (cf.\ \cite{gordy_NestedSimulationPortfolio_2010a}), but will suffice for our analysis. Then, the correlation, and subsequently the variance reduction factor, is computed for each outer path over the $1\,000$ realizations of $\ACdir$ and $\ACind$ from the inner simulation.

We observe that the crude control variate estimator generally improves the estimation significantly, in some years (2022, 2023) up to a factor of $1/5$ and very low spread between the different outer simulations. But in other years (2016, 2019, 2020) the spread is very large, ranging from no improvement up to a reduction of $1/5$ in some rare cases.

\begin{figure}[tb]
	\centering
	\includegraphics[width=1\textwidth]{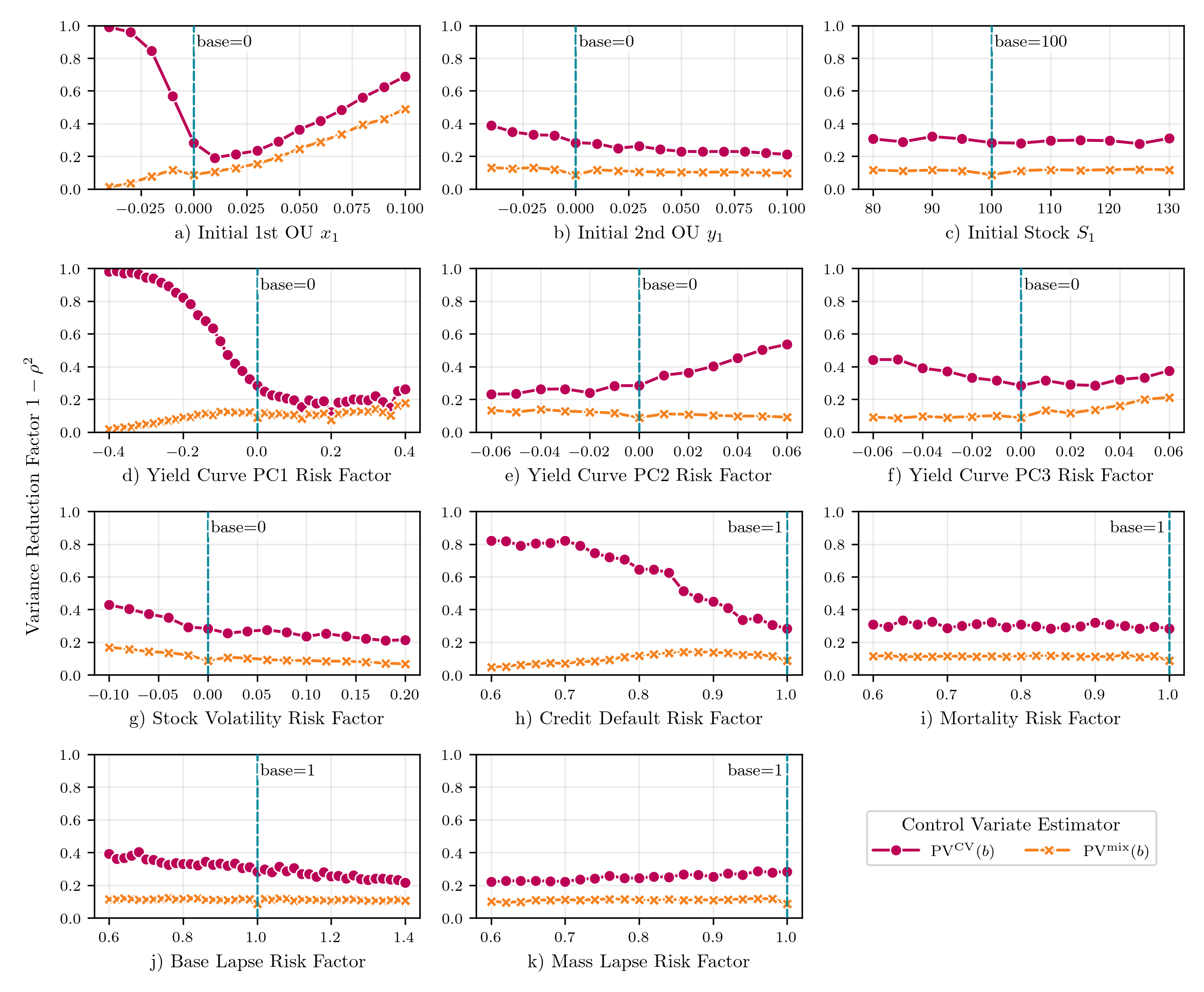}
	\caption{Impact of parameter variation on the variance reduction factor for the crude control variate estimator (with ${\scriptstyle \ACdirest - \ACindest}$ as the control) in openIRM. The parameters are used as inputs for the inner simulation of the Solvency II nested simulation after one year. Each subplot illustrates the effect of varying a single parameter, with the vertical line marking the constant base setting for all other parameters. The chosen date for the market calibration is January 3rd, 2020.}
	\label{fig:openIRM_2_per_parameter_variance_reduction}
\end{figure}

In \Cref{fig:openIRM_2_per_parameter_variance_reduction}, the variance reduction factors of $\PVcv(b)$ and $\PVmix(b)$ are depicted for openIRM with varying risk factors and all other variables kept in their respective base setting indicated by the blue vertical line.
Note that, openIRM is the most realistic internal risk model of the three benchmarks. We observe that $\PVcv(b)$ reduces the variance of our estimator by a factor of approximately $1/3$ for most parameterizations, ranging from $0.2$ in the best case up to no improvement in the worst case for the variable ranges we looked at.

Furthermore, we observe that the much more complicated $\PVmix(b)$ is able to improve the estimation substantially. This is in stark contrast to our previous observations made for the more simple Bauer's model variants. Surprisingly, $\PVmix(b)$ seems to perform especially good in parameter regions where $\PVcv(b)$ is unable to significantly reduce the variance, cf.\ Subfigure~\ref{fig:openIRM_2_per_parameter_variance_reduction}a for $x_1 < 0$ or \ref{fig:openIRM_2_per_parameter_variance_reduction}d for the PC1 risk factor being less than 0.

Recall that $\PVmix(b)$ used the arbitrary selection of $\PVmix$ with the single time steps $t \in \timeInMixed_{\mathrm{all}}$ (as well as $\PVind$) as the control. With a more intricate selection algorithm, it may be possible to improve $\PVmix(b)$ even further.

Summarizing, it seems like the superiority of either method is linked to the degree of coupling between assets and liabilities. In the MUST case, liabilities have a simple, almost bond-like dynamic. In the IS case, liabilities are much more path-dependent on asset returns. When this coupling is strong, the liabilities become volatile. The indirect method, which focuses on the net result (shareholder cash flows), might have less variance in this case because it implicitly nets out some of this shared asset-liability volatility.
In openIRM, the annual guaranteed interest rate is set to $0.25 \%$ for all years. Due to more optimistic interest rate forecasts, the simulated market quickly outperforms this guaranteed rate during the simulation. Hence, a larger part of the market returns is paid to the policyholders resulting in a stronger coupling of assets and liabilities. We believe that this might be one reason why the presented control variate estimators are successful in openIRM. Including the mixed estimators as controls seems to offset the effect of the missing coupling, although an explanation for this phenomena is future research.

\FloatBarrier

\section{Conclusion}\label{sec:conclusion}

Our initial contribution is a proof that the direct and indirect method converge to the same Available Capital, which provides a practical means of validating model implementations. We then generalize these two approaches into a mixed-method framework capable of producing $2^T$ unique estimators, where $T$ is the number of time steps in the simulation. Tests on three benchmark life insurer models show that the relative convergence speed is strongly model-dependent, and neither method can be declared universally superior.

We proposed and evaluated a set of novel control variate estimators formed from combinations of the previously derived estimators. The crude control variate, employing only the direct and indirect estimators, proved broadly effective in reducing variance across all three benchmarks. Any additional performance of more sophisticated control variates incorporating mixed estimators was found to be model-dependent. Specifically, for the openIRM benchmark, such an estimator yielded a substantial additional decrease in variance, whereas for the two Bauer models, it conferred only a slight advantage over the crude estimator.

A key advantage of our proposed estimators is their modularity. They can be used as a drop-in replacement for the standard direct estimator within existing Monte-Carlo frameworks. This compatibility means their variance-reducing properties can be compounded with those of established techniques like batching (cf.\ \cite{glasserman_MonteCarloMethods_2010}), antithetic variables (cf.\ \cite{korn_MonteCarloMethods_2010}), sequential simulations (cf.\ \cite{broadie_EfficientRiskEstimation_2011}), proxy models (cf.\ \cite{krah_LeastSquaresMonteCarlo_2020}), or sample recycling (cf.\ \cite{feng_SampleRecyclingMethod_2022}).

Several research questions remain open.
First, identifying an optimal mixed (control variate) estimator remains a challenge. The combinatorial complexity arising from the $2^T$ estimators renders a brute-force search computationally infeasible.
Second, we observed linear dependencies within certain subsets of mixed estimators but could not identify a systematic cause. Finally, the performance of these techniques on real-world portfolios, beyond our benchmarks, has yet to be assessed.

In conclusion, this work significantly expands the toolkit for practitioners estimating own funds, offering methods that can directly improve the efficiency of their existing procedures.

\subsection*{Acknowledgments}
I sincerely thank Tim Prokosch and Ralf Korn for valuable feedback that significantly improved this paper. I appreciate the thoughtful comments and suggestions by Andreas Zapp and Daniel Bauer. I extend my gratitude to Philipp Mahler and Ria Grindel for fruitful discussions.

\subsection*{Competing Interests}
Competing interests: The author declares none.

\clearpage
\printbibliography

\appendix

\counterwithin{figure}{section}

\clearpage

\section{Appendix}

\begin{table}[htbp]
	\centering
	\caption{Base setting for simulation configuration parameters used in Bauer's model.}
	\label{tab:sim_config_bauers_model}
	\begin{tabular}{ll}
		\textbf{Parameter} & \textbf{Value} \\
		\hline
		$\assetsAfter{0}$ & 110.0 \\
		$\liabsAfter{0}$ & 100.0 \\
		$\timehorizon$ & 10 \\
		$\dt$ & 0.25 \\
		$\shortrateMeanLevel$ & 0.03 \\
		$\shortrateReversionSpeed$ & 0.05 \\
		$\shortrateVola$ & 0.01 \\
		$\marketPriceOfRisk$ & 0.0 \\
		$r_0$ & 0.025 \\
		$\shortrateAssetCorr$ & 0.0 \\
		$\assetVola$ & 0.075 \\
		$\guarInt$ & 0.02 \\
		$\participationRate$ & 0.9 \\
		$\earningsFactor$ & 0.5 \\
		annual cash flows and liabilities & True \\
		measure & $\mQ$ \\
		seed & 75 \\
		$\targetRate$ & 0.05 \\
		$a$ & 0.05 \\
		$b$ & 0.30 \\
		$\alpha$ & 0.05
	\end{tabular}
\end{table}

\begin{figure}[tb]
	\centering
	\includegraphics[width=1\textwidth]{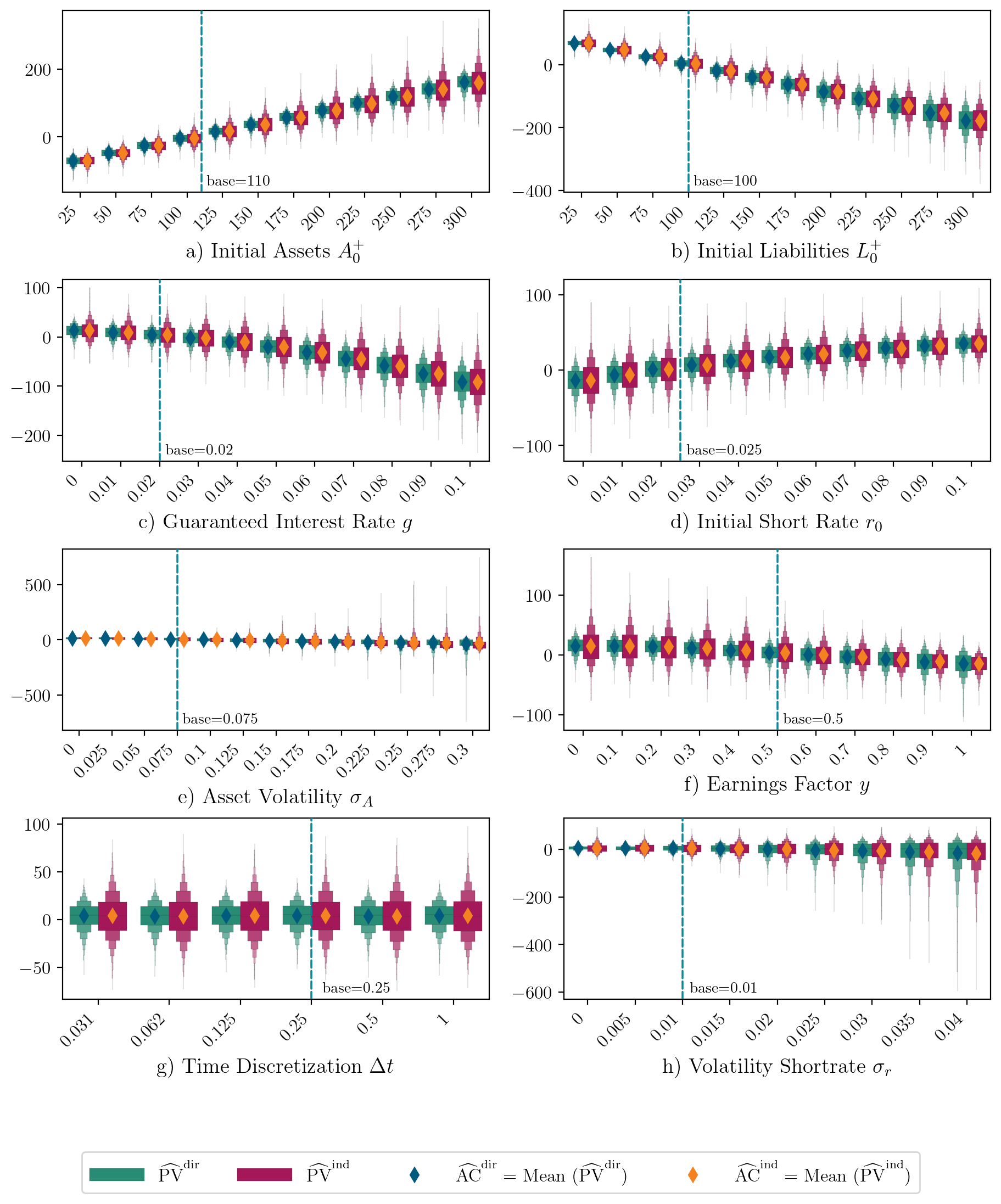}
	\caption{Distribution of $\ACdirest$ and $\ACindest$ for varying parameters in Bauer's model, MUST case.}
	\label{fig:appendix_multiplot_vary_param_MUST}
\end{figure}

\begin{figure}[tb]
	\centering
	\includegraphics[width=1\textwidth]{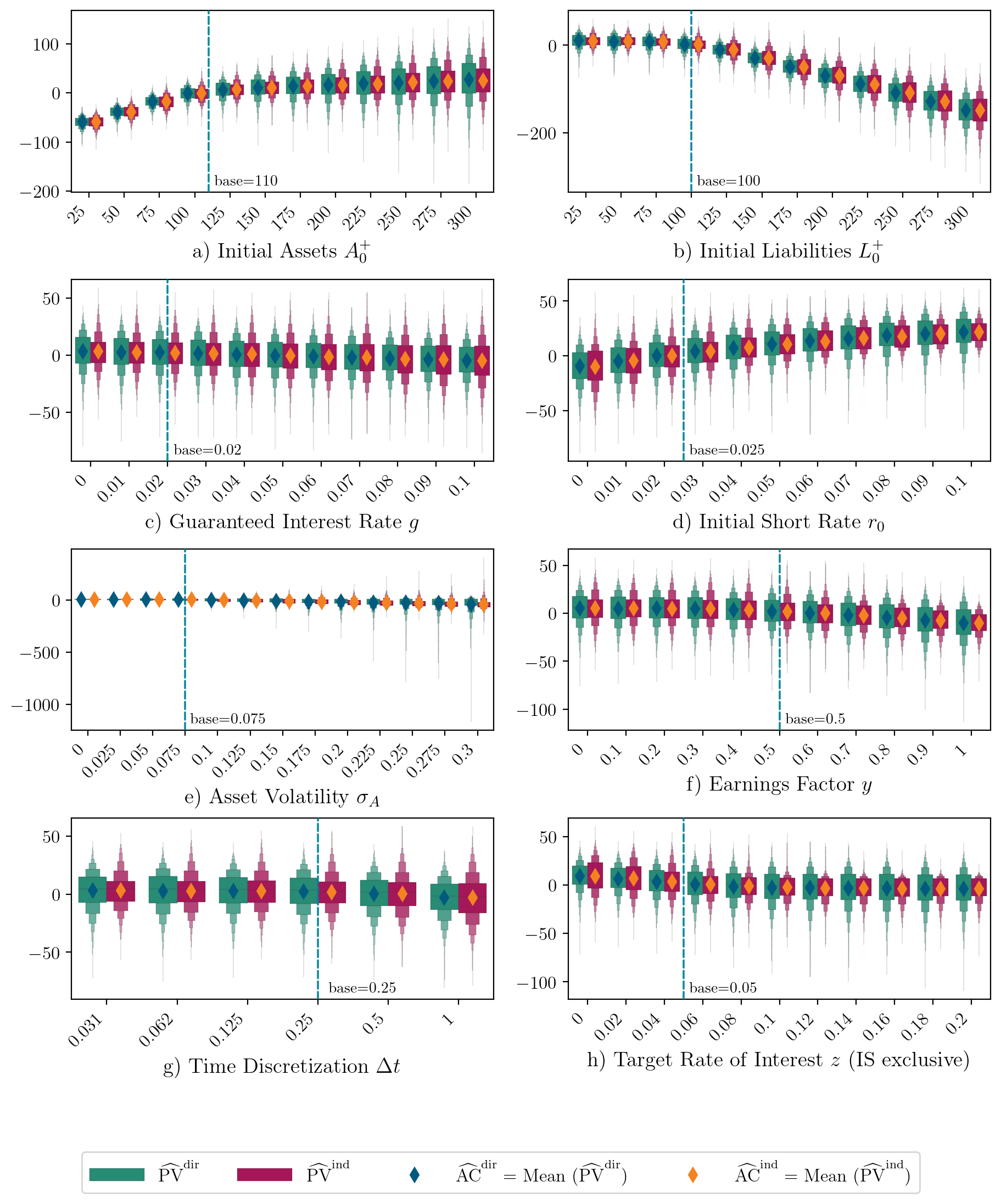}
	\caption{Distribution of $\ACdirest$ and $\ACindest$ for varying parameters in Bauer's model, IS case.}
	\label{fig:appendix_multiplot_vary_param_IS}
\end{figure}

\begin{figure}[tb]
\centering
\includegraphics[width=1\textwidth]{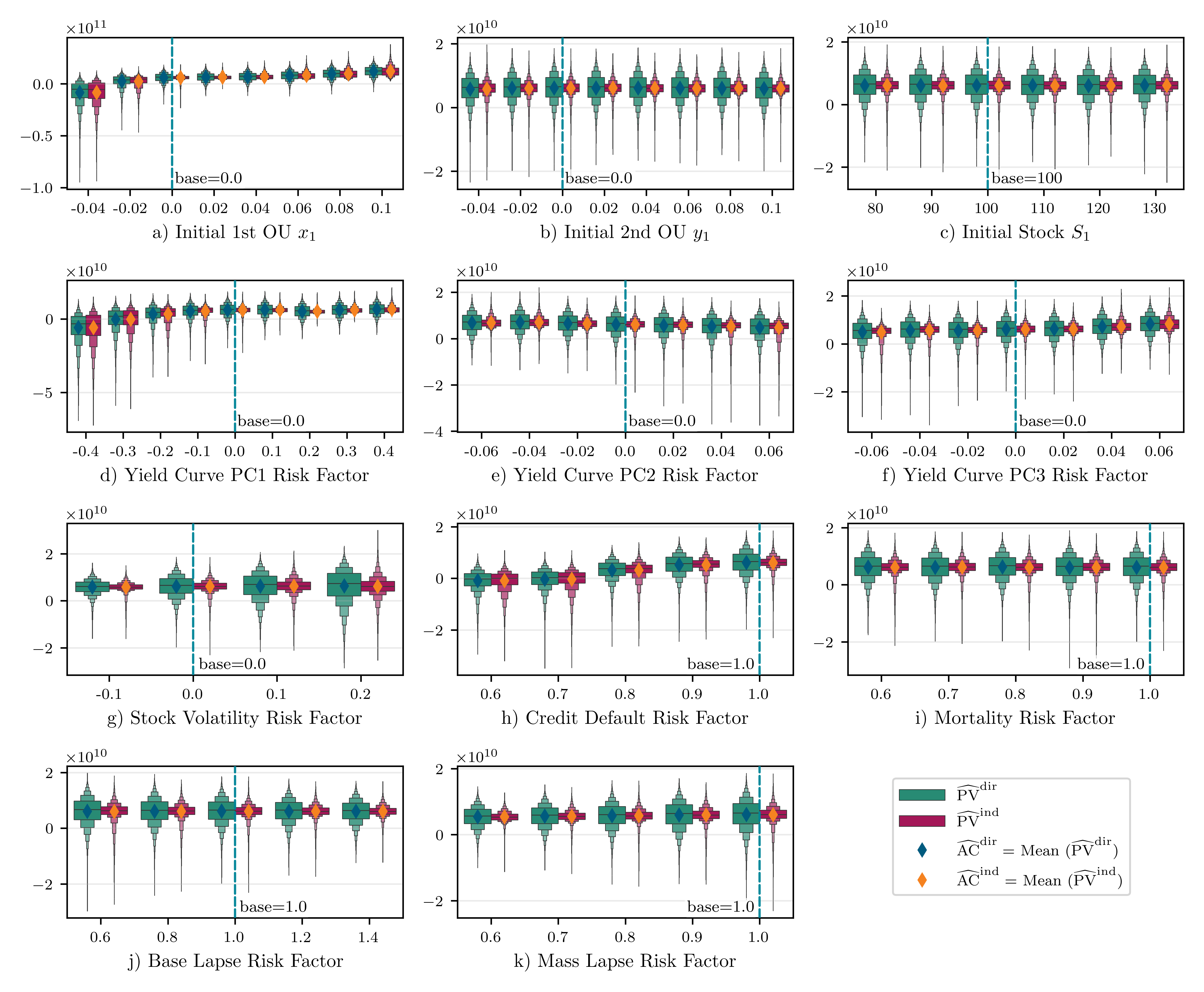}
\caption{Distribution of $\ACdirest$ and $\ACindest$ for varying parameters in openIRM.}
\label{fig:appendix_openIRM_per_parameter_boxen_direct_vs_indirect_0_outliers_removed}
\end{figure}

\begin{figure}[tb]
	\centering
	\includegraphics[width=1\textwidth]{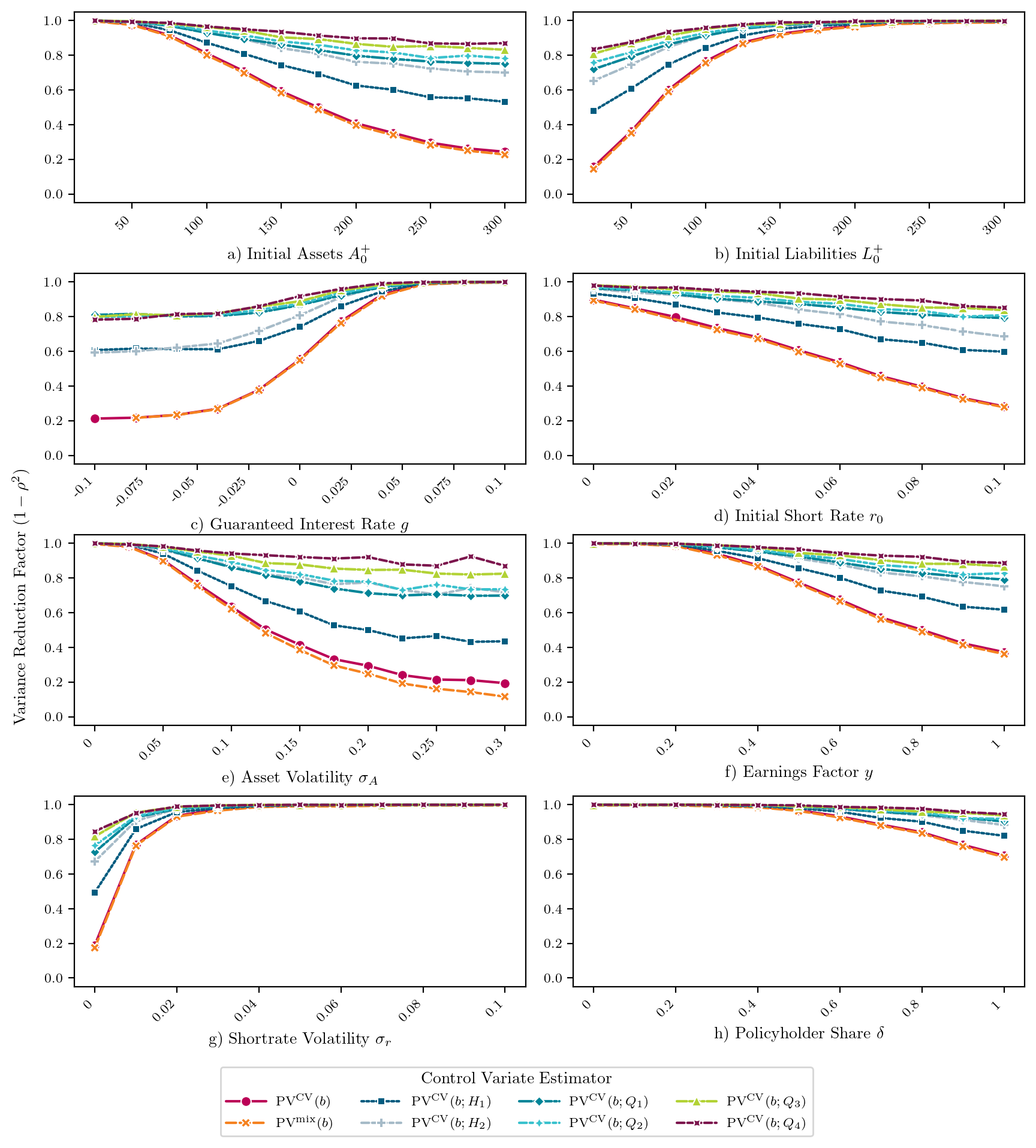}
	\caption{Variance reduction factor for various mixed control variate estimators for varying parameters in Bauer's model, MUST case. Mixed controls used: $H_i$ for $\timeInMixed$ being the first, respectively second half, and $Q_i$ for $\timeInMixed$ being the $i$-th quarter. We see that for a single control in Bauer's model MUST case, the indirect method provides the best variance reduction factor out of all mixed estimators looked at here.}
	\label{fig:appendix_multiplot_VRF_Multi_CV_Mixed_MUST}
\end{figure}

\begin{figure}[tb]
	\centering
	\includegraphics[width=1\textwidth]{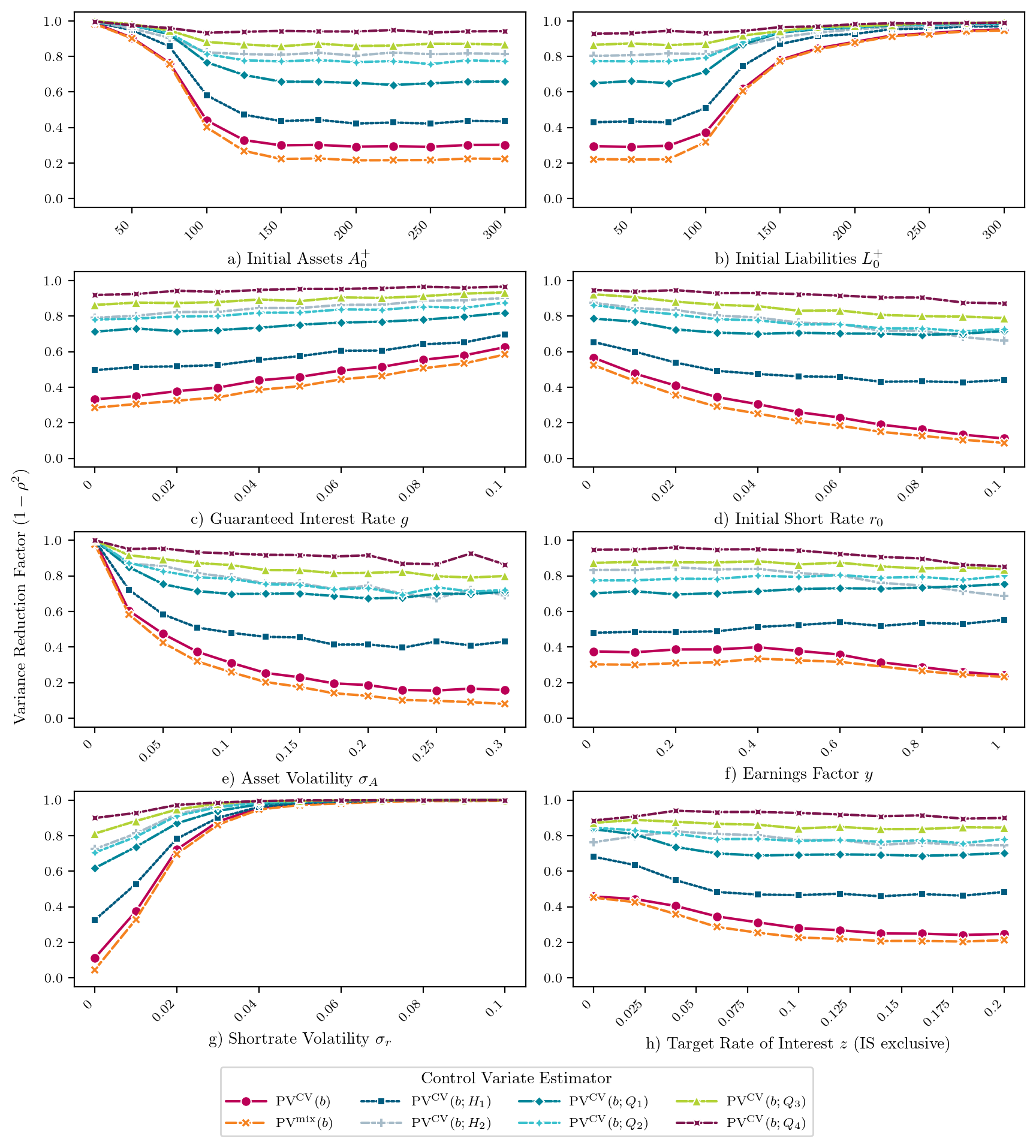}
	\caption{Variance reduction factor for various mixed control variate estimators for varying parameters in Bauer's model, IS case. Mixed controls used: $H_i$ for $\timeInMixed$ being the first, respectively second half, and $Q_i$ for $\timeInMixed$ being the $i$-th quarter. We see that for a single control in Bauer's model IS case, the indirect method provides the best variance reduction factor out of all mixed estimators looked at here.}
	\label{fig:appendix_multiplot_VRF_Multi_CV_Mixed_IS}
\end{figure}

\end{document}